\documentclass[10pt,journal,compsoc]{IEEEtran}

\ifCLASSOPTIONcompsoc
\usepackage[nocompress]{cite}
\else
\usepackage{cite}
\fi

\usepackage{soul}
\usepackage{stmaryrd}
\usepackage{balance}
\usepackage{tikz}
\usetikzlibrary{shapes}
\usetikzlibrary{calc,arrows,positioning}
\usepackage{algorithm}
\usepackage{algpseudocode}
\usepackage{amsmath}
\usepackage{booktabs}
\usepackage{amsfonts}
\usepackage[english]{babel}
\usepackage{amsthm}
\usepackage{tablefootnote}

\usepackage{graphicx}
\usepackage{subfigure}
\usepackage[colorlinks,
linkcolor=black,
anchorcolor=black,
citecolor=black
]{hyperref}

\theoremstyle{plain}
\newtheorem{thm}{Theorem}

\theoremstyle{definition}
\newtheorem{defn}{Definition}

\newcommand{\revise}{\textcolor{black}}

\newcommand{\main}{Privet}
\newcommand{\yifeng}[1]{\textsf{\color{red}{[{Yifeng: #1}]}}}
\UseRawInputEncoding

\begin{document}
\renewcommand{\algorithmicrequire}{\textbf{Input:}}  
\renewcommand{\algorithmicensure}{\textbf{Output:}} 

\renewcommand{\figurename}{Fig.}
\renewcommand{\tablename}{TABLE}

\title{Privet: A Privacy-Preserving Vertical Federated Learning Service for Gradient Boosted Decision Tables}

\author{Yifeng Zheng, Shuangqing Xu, Songlei Wang, Yansong Gao, and Zhongyun Hua
		
		\IEEEcompsocitemizethanks{
			\IEEEcompsocthanksitem Yifeng Zheng, Shuangqing Xu, Songlei Wang, and Zhongyun Hua are with the School of Computer Science and Technology, Harbin Institute of Technology, Shenzhen, Guangdong 518055, China (e-mail: yifeng.zheng@hit.edu.cn, shuangqing.xu@outlook.com, songlei.wang@outlook.com, huazhongyun@hit.edu.cn).
			\IEEEcompsocthanksitem Yansong Gao is with Data61, CSIRO, Sydney, Australia (e-mail: gao.yansong@hotmail.com).

			\IEEEcompsocthanksitem Corresponding author: Zhongyun Hua.
		}
	}

\IEEEtitleabstractindextext{%
\begin{abstract}
Vertical federated learning (VFL) has recently emerged as an appealing distributed paradigm empowering multi-party collaboration for training high-quality models over vertically partitioned datasets. Gradient boosting has been popularly adopted in VFL, which builds an ensemble of weak learners (typically decision trees) to achieve promising prediction performance. Recently there have been growing interests in using decision table as an intriguing alternative weak learner in gradient boosting, due to its simpler structure, good interpretability, and promising performance. In the literature, there have been works on privacy-preserving VFL for gradient boosted decision trees, but no prior work has been devoted to the emerging case of decision tables. Training and inference on decision tables are different from that the case of generic decision trees, not to mention gradient boosting with decision tables in VFL. In light of this, we design, implement, and evaluate Privet, the first system framework enabling privacy-preserving VFL service for gradient boosted decision tables. Privet delicately builds on lightweight cryptography and allows an arbitrary number of participants holding vertically partitioned datasets to securely train gradient boosted decision tables. Extensive experiments over several real-world datasets and synthetic datasets demonstrate that Privet achieves promising performance, with utility comparable to plaintext centralized learning.

\end{abstract}

\begin{IEEEkeywords}
Vertical federated learning service, multi-party collaboration, gradient boosting, decision table, privacy preservation
\end{IEEEkeywords}}

\maketitle

\IEEEdisplaynontitleabstractindextext

\IEEEpeerreviewmaketitle

\ifCLASSOPTIONcompsoc
\IEEEraisesectionheading{\section{Introduction}\label{sec:introduction}}
\else
\section{Introduction}
\label{sec:introduction}
\fi

Federated learning (FL) has recently emerged as a fascinating distributed machine learning paradigm that greatly empowers multi-party collaboration for mining value over data federation \cite{bonawitz2017practical,TKDE-FL1,TKDE-FL2, LU1}. It allows distributed individual training datasets to be kept locally, and only intermediate outputs from the training algorithm are shared out for aggregation.
%
%
According to how data is distributed among the participants in FL, there are two types of FL: horizontal federated learning (HFL) \cite{li2020practical,maddock2022federated} and vertical federated learning (VFL) \cite{tist,fang2021large}.
HFL addresses the scenario where the participants share the same feature space but hold disjoint sets of samples/instances, which generally suits the case that participants are individual customers.
In contrast, VFL targets the scenario where each participant has the same set of samples/instances yet owns data for different features, which is more common when the participants are business organizations/enterprises.
For example, as illustrated in Fig. \ref{fig:vfl}, the participants hold datasets that have the same row indexes (corresponding to the same set of instances) but different non-overlapping column indexes (corresponding to different features).
In this paper, we focus on the VFL setting, which has received increasing attentions in the collaboration of different business organizations/enterprises in recent years \cite{fang2021large,fu2021vf2boost}.

For model training in the VFL setting, the gradient boosting technique has received wide attentions \cite{cheng2021secureboost,fang2021large,tist,tian2020federboost,fu2021vf2boost} and has seen popular adoption for empowering a wide range of fields, such as web search ranking, online advertisement, and fraud detection \cite{tyree2011parallel,he2014practical,dhieb2019extreme}.
Gradient boosting builds an ensemble of weak learners, which are typically (generic) decision trees, to achieve promising prediction performance.
%
%
While decision tree is usually used as the weak learner in gradient boosting, in recent years there has been a fast-growing trend to use decision table \cite{kohavi1995oblivious} as an intriguing alternative \cite{bdt,prokhorenkova2018catboost,dato2016fast, HancockK20a}.
Many works \cite{bdt,prokhorenkova2018catboost,MonoForest,ltr2} have shown that gradient boosted \emph{decision tables} yields promising performance on various tasks and achieves great inference efficiency over generic decision trees.
In addition, some famous open-source gradient boosting libraries \cite{lightgbm,prokhorenkova2018catboost} have also recently provided the support for using decision table as the weak learner in gradient boosting.

As demonstrated in Fig. \ref{fig:tree-comparison}, a $D$-dimensional decision table at a high level consists of $D$ Boolean tests and $2^D$ output values. It can also be treated as a \emph{special} full binary decision tree, called \emph{oblivious tree}.
In contrast with generic decision tree which has different Boolean tests at different internal nodes at the same level, the internal nodes at the same level of an oblivious tree share the same Boolean test defined with the same feature and threshold.
Despite the similarly equivalent tree structure, it is worth noting that the algorithm for training oblivious tree is different from that for generic decision tree \cite{breiman1984classification,quinlan2014c4}.
\revise{
Specifically, decision trees are typically trained through recursive algorithms \cite{breiman1984classification,quinlan2014c4}, while decision tables are trained through iterative algorithms following the top-down construction \cite{kohavi1995oblivious,bdt}.
Given tree depth $D$, the shape of a generic decision tree is uncertain because it needs to process samples associated with the current node to determine whether to split this node. 
In contrast, we can not recursively build a decision table because all the samples in the dataset need to be processed to select the optimal split for each level of the decision table. Besides, given depth $D$, the shape of an oblivious tree is fixed and the number of operations like node splitting and output value calculation is also fixed. 
In addition, the inference process on an oblivious tree is also different from that on a generic decision tree\cite{bdt,dato2016fast} (see Section \ref{sec:decision_table} for more detailed discussion).
}

\begin{figure}[!t]
	\centering
	\includegraphics[scale=0.5]{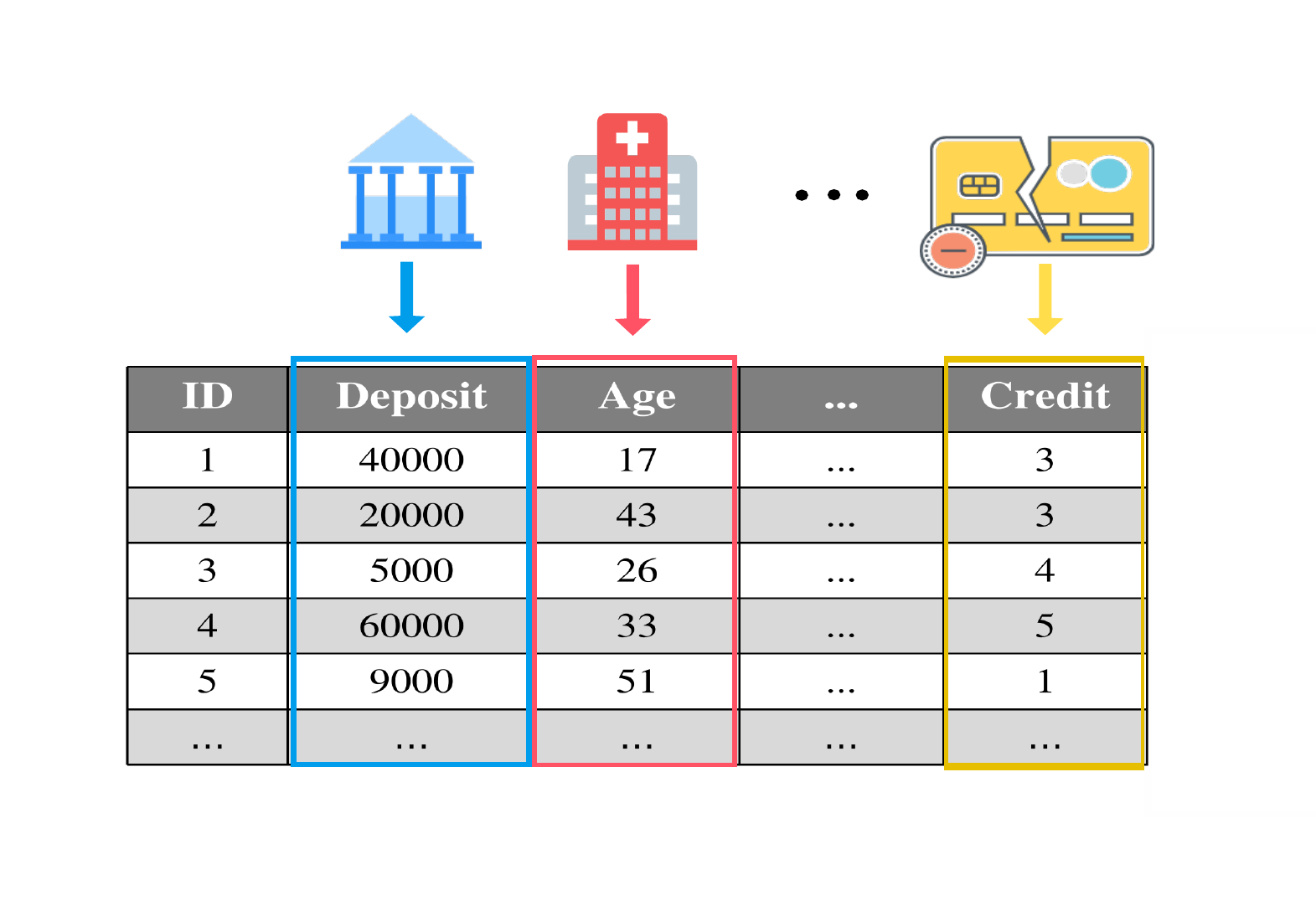}
	\caption{Illustration of data partitioning in the VFL setting.}
	\label{fig:vfl}
	\vspace{-10pt}
\end{figure}

In the literature, while there have been several studies on privacy-preserving VFL with gradient boosted decision trees (GBDT) \cite{cheng2021secureboost,fang2021large,tist}, no prior work has explored privacy-preserving VFL with gradient boosted decision tables.
As mentioned above, even training and inference on decision table are different from the case of generic decision tree, not to mention gradient boosting with decision table as the weak learner in the VFL setting.
Therefore, these prior works cannot be directly applied to support privacy-preserving training and inference of gradient boosted decision tables in VFL.
%
%
In addition, it is noted that these prior works are also confronted with limitations such as exposing sensitive intermediate results (e.g., sum of gradients) \cite{cheng2021secureboost,tist}, supporting training only among two participants \cite{fang2021large} (see Section \ref{sec:related-work} for more detailed discussion).

In light of the above, we propose {\main}, which, to our best knowledge, is the first system framework enabling privacy-preserving VFL service for training \emph{gradient boosted decision tables} over distributed datasets.
{\main} ambitiously supports an arbitrary number of participants to collaboratively train gradient boosted decision tables, while allowing them to keep their data locally and offering strong protection on the sensitive intermediate outputs throughout the training process.
%
%
{\main} builds on lightweight secret sharing techniques to develop customized protocols securely realizing the key components required by training gradient boosted decision tables in the VFL setting.

Specifically, through an in-depth examination on the training process of gradient boosted decision tables, we manage to decompose the holistic secure design in the VFL setting into the design of a series of secure components run in a distributed manner among the participants, including secure node splitting, secure Sigmoid evaluation, secure discretization, and secure distributed decision table inference.
The delicate synergy of these secure components leads to the holistic protocol of {\main} for privately training gradient boosted decision tables in the VFL setting.
Through the customized secure protocol, {\main} outputs gradient boosted decision tables that are distributed among the involved participants, where each participant only holds a part of the model.
Subsequently, secure inference on the ensemble of learned decision tables can also be well supported in a distributed manner among the participants.
%
We implement and evaluate {\main}'s protocols extensively over several real-world datasets as well as synthetic datasets.
The results demonstrate that {\main} presents promising performance in computation and communication.
Meanwhile, the utility of the trained models in {\main} is comparable to that in the plaintext centralized learning setting.

We highlight our contributions as follows.

\begin{itemize}
	\item \revise{We present {\main}, which, to our best knowledge, is the first system framework enabling privacy-preserving VFL service for gradient boosted decision tables. Privet allows an arbitrary number of participants holding vertically partitioned distributed datasets to securely train gradient boosted decision tables in a distributed manner, offering strong protection for sensitive individual data as well as for intermediate outputs.}
	
	\item \revise{We devise a series of tailored secure components based on lightweight secret sharing techniques that run in a distributed manner among multiple participants with promising efficiency and utility, catering for the computation required by securely training gradient boosted decision tables in the VFL setting.} 
	
	\item \revise{We make an implementation of the proposed protocols and conduct an extensive evaluation over three real-world public datasets and three synthetic datasets. 
	The experiment results demonstrate that {\main} has promising performance, achieving model utility comparable to plaintext centralized learning.}
\end{itemize}

The rest of this paper is organized as follows. Section \ref{sec:related-work} discusses the related work. Section \ref{sec:preliminaries} introduces some preliminaries. Section \ref{sec:overview} gives a system overview. Section \ref{sec:design} presents the design of {\main}. The security analysis is presented in Section \ref{sec:security-analysis}, followed by the experiments in Section \ref{sec:experiments}. Section \ref{sec:conclusion} concludes the whole paper.

\section{Related Work}
\label{sec:related-work}

\noindent \textbf{Securely learning gradient boosted decision trees under HFL.} Due to the problems of data isolation and data privacy, FL has emerged as a new privacy-preserving machine learning paradigm. Several existing works \cite{li2020practical,leung2019towards,maddock2022federated} have been focused on privacy-preserving gradient boosted decision trees (GBDT) under the HFL setting, which assume that data are horizontally partitioned between participants. 
Among them, the work \cite{maddock2022federated} rely on use of secure aggregation and differential privacy to provide a privacy guarantee. The work \cite{leung2019towards} leverages secure hardware\cite{hasp} to build private GBDT under HFL. 
Different from these works, our work targets privacy-preserving gradient boosting systems under the VFL setting.

\noindent \textbf{Securely learning gradient boosted decision trees under VFL.}
To cater for the need to collaboratively build models between different organizations that hold data on the same set of samples but for different features, VFL has received increasing attention in recent years.
The works \cite{cheng2021secureboost,fang2021large,tist} consider vertical federated gradient boosted decision trees, which are most related to ours. 
In particular, SecureBoost \cite{cheng2021secureboost} is the first work on privacy-preserving GBDT over vertically partitioned data, which uses homomorphic encryption to preserve data privacy. 
However, it has limited security guarantee because intermediate information (e.g., the sum of gradients in a bucket) is revealed during the training process. 
Moreover, homomorphic encryption involves heavy cryptographic operations and requires large memory, which results in low training efficiency. 
The works\cite{fang2021large,tist} improve SecureBoost \cite{cheng2021secureboost} in terms of efficiency via multi-party computation (MPC) techniques. 
Specifically, the work \cite{fang2021large} proposes a secure GBDT system leveraging the additive secret sharing technique \cite{mohassel2017secureml}.
However, their proposed system is only designed for the two-party setting.
Xie \textit{et al.} \cite{tist} deal with the issue to support secure multi-party training.
However, since they adopt large-scale matrix multiplication in the secret sharing domain to discretize secret-shared gradients into buckets, their scheme requires more communication and computation overhead compared to \cite{fang2021large}. 
Moreover, the design in \cite{tist} has notable privacy leakages, e.g., the intermediate inference results of all training samples are leaked to the participant who holds the label set because it relies on this participant to conduct inference.

We also note that all these works \cite{cheng2021secureboost,fang2021large,tist} are aimed at supporting secure training and inference for gradient boosting with generic decision trees under VFL.
\revise{In recent years, the gradient boosted decision table technique has seen rapidly growing adoption in various applications, such as learning to rank (LTR) \cite{dato2016fast,ltr2,ltr3}, recommendation systems\cite{bdt,dhar2020effective}, and medical diagnosis \cite{gbdt_diabetes,gbdt_diabetes2}.}
Although the training of decision tree and decision table has some similarities, e.g., both of them need permutation protocols, their learning algorithms are different inherently.
Thus the works \cite{cheng2021secureboost,fang2021large,tist} cannot directly support secure gradient boosting over decision tables under VFL. 
In comparison with them, {\main} focuses on securely supporting privacy-preserving VFL for gradient boosted decision tables. In addition, {\main} departs from them by achieving comparable utility to plaintext, concealing intermediate information for strong privacy, and supporting an arbitrary number of participants.

\noindent\textbf{Secure decision tree learning supporting both horizontally and vertically partitioned data.} There are some works \cite{popets,deforth2021xorboost,adams2022privacy} which can support secure decision tree learning on both horizontally and vertically partitioned data in an \emph{outsourcing} setting.
Specifically, the work \cite{popets} considers a setting where data owners secret-share all their data among three servers and designs a protocol to enable the three servers to securely perform an adapted C4.5 decision tree learning algorithm.
%
%
The work \cite{adams2022privacy} proposes protocols to train decision trees for the \emph{Random Forest} model, which similarly considers a setting where the data owners secret-share all their data among two extra non-colluding computing parties. 
In \cite{deforth2021xorboost}, Deforth \emph{et al.} focus on building private gradient boosted decision trees and consider a scenario where data owners secret-share their data among a set of computing parties which may also be an extra set of servers.
In contrast with these works that outsource the data and computation, {\main} does not require such an extra set of non-colluding servers which may not be an easy assumption to meet in practice.
Meanwhile, {\main} allows the raw data of each participant to stay local throughout the whole training process, fitting the salient feature of FL.


\section{Preliminaries}
\label{sec:preliminaries}

\subsection{Decision Table}
\label{sec:decision_table}
Consider a dataset $\mathcal{D}$ consisting of $N$ samples $\{\mathbf{x}_i,y_i\}$ for $i=0,\cdots,N-1$, where $\mathbf{x}_i=(x_{i1},\cdots,x_{iJ})$ is a $J$-dimensional tuple and $y_i$ is the label of the $i$-th sample. 
The $j$-th element of $\mathbf{x}_i$ is the value of an input attribute $X_j$.
%
%
A $D$-dimensional decision table consists of $D$ Boolean tests and $2^{D}$ output values. A Boolean test is of the form $X_j<t$, which outputs $1$ if the $j$-th element in a given input tuple is less than a threshold $t$ and $0$ otherwise.


As illustrated in Fig. \ref{fig:tree-comparison}, a $D$-dimensional decision table is equivalent to a full binary tree with $D+1$ levels, where each internal node from the $0$-th level (for the root node) to the $(D-1)$-th level has a Boolean test; each edge is assigned the outcome of its source node's test and each leaf node at the $D$-th level is associated with an output value. Such equivalent tree is called \emph{oblivious tree}, because all internal nodes at the same level share the same test, as opposed to generic decision trees that have different tests at the same level.
More specifically, the test at the $d$-th level of an oblivious tree could be represented as $F_d<t_d$, where $d\in[0,D-1]$, the split feature $F_d\in\{X_1,\cdots,X_J\}$, and $t_d$ is the split threshold.
The special structure of oblivious tree results in its different training and inference methods from non-oblivious trees like CART \cite{breiman1984classification}.
In \cite{kohavi1995oblivious}, Kohavi \textit{et al.} first introduce a top-down construction of oblivious  trees and use information gain as the evaluation metric to find the optimal test at each level. Different evaluation metrics are used in later studies, like mean squared error (MSE)\cite{bdt} and Newton's method\cite{prokhorenkova2018catboost}.

\begin{algorithm}[!t]
	\caption{Training an Oblivious Tree} 
	\label{algo:plaintext-oblivious-tree}
	\begin{algorithmic}[1]
		\Require
		A training dataset $\mathcal{D}$.
		\Ensure
		An oblivious decision tree having $D$ tests and $2^D$ output values.
		\State $\mathcal{S}^0 = \{\mathcal{D}\}.$ \label{1-line1}
		\For{$d\in [0,D-1]$} 
		\State $\mathcal{S}^{d+1} = \{\} $.
		\State Optimal test $F_d < t_d$ $\leftarrow$ $\mathsf{find\_split}$.\label{1-line5}
		
		\For{$\mathcal{V}$ in $\mathcal{S}^d$}
		\State Split $\mathcal{V}$ into $\mathcal{V}_{F_d<t_d}$, $\mathcal{V}_{F_d\geq t_d}$  according to the opti-   
		\Statex \quad\quad\:mal test and add these two sets to $\mathcal{S}^{d+1}$. \label{1-line6}
		\State Create a node for each set in $\mathcal{S}^{d+1}$ and connect  \label{1-line7}
		\Statex \quad\quad\:it to its parent node.
		\EndFor				
		\EndFor
		\State Calculate output values for the $2^D$ leaf nodes at the $D$-th level, respectively.
	\end{algorithmic}
\end{algorithm}

We follow the top-down construction in \cite{kohavi1995oblivious,bdt} to train oblivious trees. Algorithm \ref{algo:plaintext-oblivious-tree} shows the process of training an oblivious tree, which produces $D$ tests and $2^D$ output values.
The learning algorithm starts from the $0$-th level and builds an oblivious tree level by level iteratively. 
Given a test $X_{j}<t$, we define $\mathcal{D}_{X_{j}<t}=\left\{(\mathbf{x}, y) \in \mathcal{D} \mid \mathbf{x}\left(X_{j}\right)<t\right\}$, $\mathcal{D}_{X_{j}\geq t}=\mathcal{D} \backslash \mathcal{D}_{X_{j} < t}$.
We also apply this notation to subsets $\mathcal{V}\subseteq \mathcal{D}$. Let $\mathcal{S}^l$ denote the set of $\mathcal{D}$'s subsets at the $l$-th level, where $l\in[0,D]$.
At level $0$, the training dataset $\mathcal{D}$ is associated with the root node and $\mathcal{S}^0 = \{\mathcal{D}\}$ (line~\ref{1-line1}). Once an optimal test $F_0<t_0$ at this level is found through the routine $\mathsf{find\_split}$ (line~\ref{1-line5}), $\mathcal{D}$ is partitioned into two subsets $\mathcal{D}_{F_0<t_0},\mathcal{D}_{F_0\geq t_0}$ according to it. After that, $\mathcal{D}_{F_0<t_0},\mathcal{D}_{F_0\geq t_0}$ are added to $\mathcal{S}^1$ ($\mathcal{S}^1=\{\mathcal{D}_{F_0<t_0},\mathcal{D}_{F_0\geq t_0}\}$) and a new level is created (lines~\ref{1-line6}-\ref{1-line7}).

\begin{figure}[!t]
	\centering
	\includegraphics[scale=0.30]{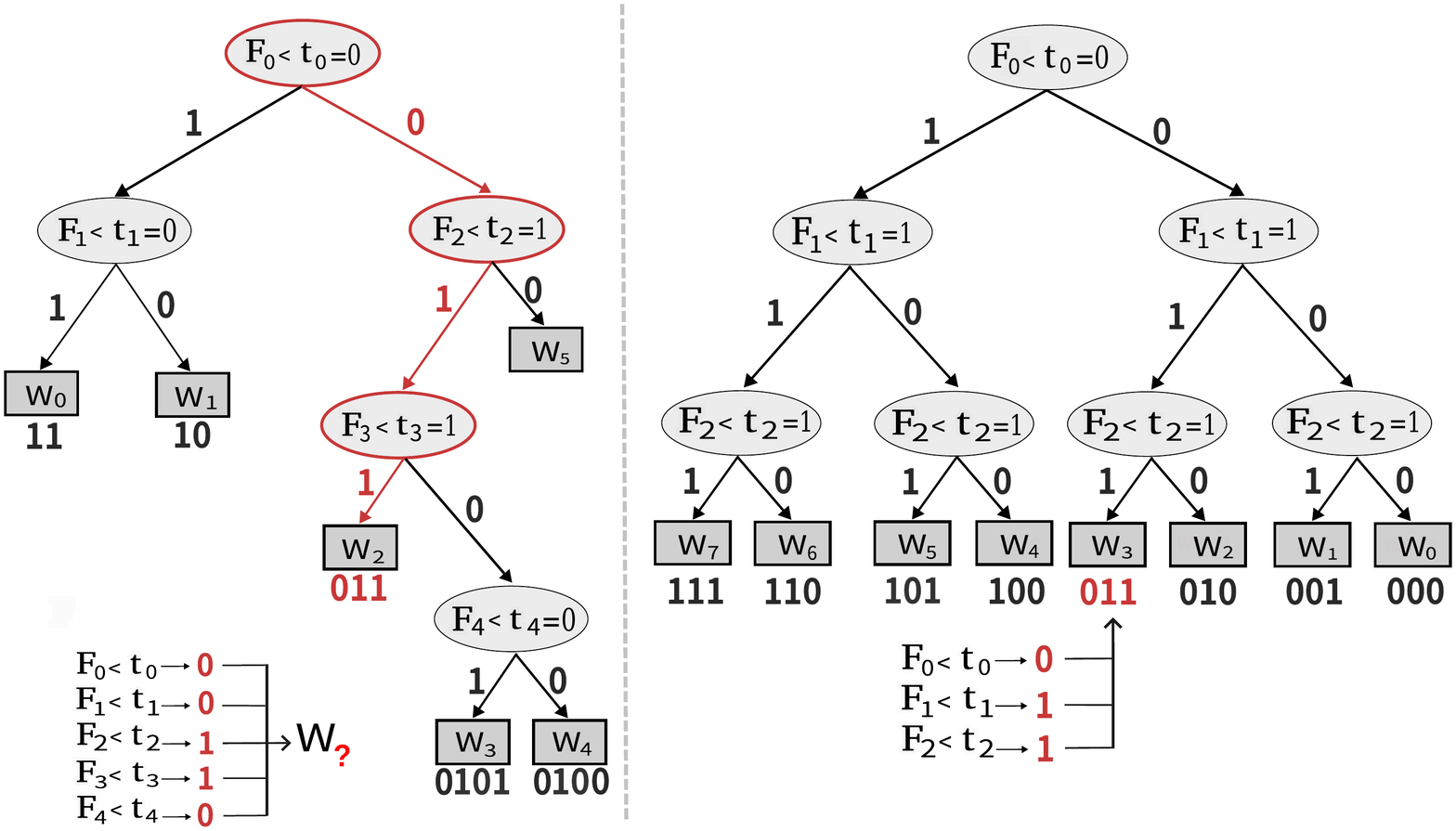}
	\caption{\revise{Comparison of a generic decision tree and an oblivious tree in inference.}}
	\label{fig:tree-comparison}
	\vspace{-15pt}
\end{figure}

At level $1$, an optimal test $F_1<t_1$ is found and $\mathcal{D}_{F_0<t_0},\mathcal{D}_{F_0\geq t_0}$ are each partitioned into two subsets according to $F_1<t_1$. The same procedure is repeated until all the $D$ tests are learned. In this way, the tree structure is kept full and symmetric, and we have $\vert \mathcal{S}^l \vert=2^l$ at level $l$, where each set in $\mathcal{S}^l$ is associated with a node at this level.
When reaching the $D$-th  level, the output values will be calculated for the leaf nodes. Finally, an oblivious tree composed of $D$ tests and $2^D$ output values is learned.

The optimal test at each level is found via the routine $\mathsf{find\_split}$ by evaluating the candidate tests. Evaluation of the candidate tests can be made through different metrics. In {\main}, we follow the popular second-order approximation method \cite{xgboost,prokhorenkova2018catboost} to evaluate tests because the decision tables in our work are trained sequentially for a gradient boosting system. Besides, the output values of decision tables can also be calculated following the gradient boosting theory, which will be introduced shortly in Section \ref{gradient-boosting}.

\revise{
As presented in Algorithm \ref{algo:plaintext-oblivious-tree}, training a decision table (oblivious tree) is an iterative process, while decision trees are typically trained through recursive algorithms\cite{breiman1984classification,quinlan2014c4}. 
Given tree depth $D$, the shape of a generic decision tree is uncertain because it needs to process samples associated with the current node to judge whether to split this node. 
However, the shape of an oblivious tree is predetermined at a given dimension $D$. To select the optimal split at each level, all the samples in the dataset are required to be processed. 
Besides, the number of operations involved in training an oblivious tree, such as $\mathsf{find\_split}$ and output value calculation, is fixed.
}

Decision table outperforms generic decision tree in inference efficiency significantly. As illustrated in Fig. \ref{fig:tree-comparison}, each leaf node of an oblivious tree (the right sub-figure in Fig. \ref{fig:tree-comparison}) corresponds to a Boolean sequence and the comparisons required by $D$ tests could be parallelized. In contrast, inference in a regular decision tree is made by traversing the tree from the root node to a leaf node, which means the direction of the inference path after the current node depends on the test result of this node. 
\revise{
Note that while the evaluation of each decision node in generic decision tree inference can be parallelized, it is still necessary to traverse the tree from the root node \emph{sequentially} so as to identify the correct leaf node that produces the inference result.
For example, as shown in the left sub-figure in Fig. \ref{fig:tree-comparison}, even if we parallelize the evaluation of each decision node, i.e., we obtain the sequence of test results $[0,0,1,1,0]$ by evaluating the $0$-th split to the $4$-th split simultaneously, we cannot directly identify which leaf node is finally chosen using $[0,0,1,1,0]$. 
On the contrary, decision table inference is free of such sequential traversal \cite{bdt,prokhorenkova2018catboost,dato2016fast}. 
As illustrated in the right sub-figure in Fig. \ref{fig:tree-comparison}, once the sequence comprised of Boolean test result at each level is obtained, the inference result can be immediately obtained because this Boolean sequence is also the identifier of a leaf node. 
}

\revise{
Additionally, it is noted that in gradient boosting systems, the number, size, and depth of generic decision trees are not necessarily smaller than decision tables when achieving the same accuracy because they are both weak learners and only require weak predictability.
As reported in prior work \cite{bdt}, compared with gradient boosted decision trees with the number of trees $T=50$ and tree depth $D=7$, gradient boosted decision tables only requires depth $D=6$ given the same number of oblivious trees $T=50$ to achieve similar accuracy performance.
Furthermore, it is noted that with the same depth $D$, a decision table only needs storage of $D$ decision nodes (one for each level), while a generic decision tree may require storage of up to $2^D-1$ decision nodes \cite{bdt,dato2016fast}.
}

\subsection{Gradient Boosted Decision Tables}
\label{gradient-boosting}

A gradient boosting system is built by training a set of weak learners sequentially based on the boosting algorithm\cite{xgboost,friedman2000additive}. For the given dataset $\mathcal{D}=\{\mathbf{x}_i,y_i\}^{N-1}_{i=0}$, a gradient boosting system sums the inference results of $T$ weak learners to produce the ultimate inference result for the $i$-th sample\cite{friedman2000additive}: $\hat{y}_{i}^{(T)}=\sum_{t=1}^{T} f_{t}(\mathbf{x}_i)$, where $f_{t}$ corresponds to the model of the $t$-th weak learner. In gradient boosted decision tables\cite{bdt,prokhorenkova2018catboost}, $f_{t}$ corresponds to a decision table. A given sample will be classified into the leaf nodes in the decision tables according to the tests in them. Its ultimate inference result is calculated by summing up the output values associated with the corresponding leaf nodes.

The essence of gradient boosting algorithm comes from how it boosts the weak learners sequentially. After training $t-1$ weak learners, the $t$-th model $f_t$ is needed to be trained and added to minimize the following objective function\cite{xgboost}:
\begin{equation}
\begin{aligned}
\mathcal{L}^{(t)} & =\sum_{i=0}^{N-1} l(y_{i}, \hat{y}_{i}^{(t)})+\Omega(f_{t})\\
& =\sum_{i=0}^{N-1} l(y_{i}, \hat{y}_{i}^{(t-1)}+f_{t}(\mathbf{x}_{i}))+\Omega(f_{t}), \\
\end{aligned}
\notag
\end{equation}

\noindent where $l$ is a twice differentiable convex loss function that takes $y_i$,  $\hat{y}_{i}^{(t)}$ as input, and outputs the loss. The regularization term  $\Omega(f_{t})$ is set following \cite{xgboost}.  Friedman \textit{et al.} \cite{friedman2000additive} use second-order approximation to quickly approximate the objective function: 
\begin{equation}
\begin{aligned}
\tilde{\mathcal{L}}^{(t)} \simeq 
& \sum_{i=0}^{N-1}[(l(y_{i}, \hat{y}^{(t-1)})+g_{i} f_{t}(\mathbf{x}_{i})+\frac{1}{2} h_{i} f_{t}^{2}(\mathbf{x}_{i})]+\Omega(f_{t}), \\
\end{aligned}
\label{eq:second-order}
\end{equation}
where $g_{i}=\partial_{\hat{y}_{i}^{(t-1)}} {l}(y_{i}, \hat{y}_{i}^{(t-1)})$, $h_{i}=\partial_{\hat{y}_{i}^{(t-1)}}^{2} {l}(y_{i}, \hat{y}_{i}^{(t-1)})$ are the first and second-order gradients of the $i$-th sample. Typically, for regression problems, MSE is used as the loss function and the gradients are calculated as follows: $g_i=\hat{y}_i-y_i$ and $h_i=1$\cite{bdt,tist}. When the problem is classification, a common choice is logistic loss and the gradients are calculated as follows: $g_i=p_i-y_i$ and $h_i=p_i\times(1-p_i)$, where $p_i=\mathsf{Sigmoid}(\hat{y}_i)$\cite{xgboost_para}. 
For a value ${x} \in \mathbb{R}$, the Sigmoid function is: $\mathsf{Sigmoid}({x})=1/(1+e^{-x})$.
For the leaf node $k$, which is associated with a subset $\mathcal{V}^k\subset \mathcal{D}$, we define $\mathcal{I}^k=\{i\vert(\mathbf{x}_i,y_i)\in \mathcal{V}^k\}$ as its index set. This notation is also used to denote the index set associated with the internal node, e.g., we write $\mathcal{I}^q$ for node $q$.
Then, after removing the constant terms, Eq. \ref{eq:second-order} can be rewritten as\cite{xgboost}:
\begin{equation}
\begin{aligned}
\tilde{\mathcal{L}}^{(t)}  
&  =\sum_{k=0}^{L-1}[(\sum_{i \in \mathcal{I}^{k}} g_{i}) w_{k}+\frac{1}{2}(\sum_{i \in \mathcal{I}^{k}} h_{i}+\lambda) w_{k}^{2}]+\gamma L\\
\end{aligned}
\label{eq:rewrite}
\end{equation}
\noindent where $w_{k}$ is the output value associated with the leaf node $k$, $L$ is the number of leaf nodes in the tree, and $\lambda,\gamma$ are hyper-parameters to control the regularization.
When the tree stops growing, $w_{k}$ and the minimum loss of the current tree are calculated by \cite{xgboost}:
\begin{equation}
\label{eq:weight}
w_{k}=-\frac{\sum_{i \in \mathcal{I}^{k}}g_i}{\sum_{i \in \mathcal{I}^{k}}h_i+\lambda},
\end{equation}
\begin{equation}
\label{eq:tree-structure}
\begin{aligned}
\tilde{\mathcal{L}}^{(t)}&=-\frac{1}{2} \sum_{k=0}^{L-1} \frac{(\sum_{i \in \mathcal{I}^{k}}g_i)^{2}}{\sum_{i \in \mathcal{I}^{k}}h_i+\lambda}+\gamma L,
\end{aligned}
\end{equation}

\noindent Eq. \ref{eq:tree-structure} can be used as the impurity function for evaluating the tests. In {\main}, we follow the above theory to find optimal tests in decision table. Suppose we have learned $b$ tests from the level $0$ to the level $(b-1)$ of a $D$-dimensional decision table and we need to find an optimal test at level $b$.
The nodes at the $b$-th level are numbered from $0$ to $2^{b}-1$ and the $q$-th node is associated with an index set $\mathcal{I}^q$.
A candidate test $X_j<t$ will split the $2^{b}$ nodes at this level into $2^{b+1}$ nodes. Among all the candidate tests, the optimal test is the test that has the minimum score. The definition of score is defined as\cite{catboost_score}:
\begin{equation}
\label{eq:score}
\begin{aligned}
Score = \sum_{q=0}^{2^{b}-1}\left(\mathcal{P}_{\mathcal{I}_{L}^{q}}+\mathcal{P}_{\mathcal{I}_{R}^{q}} \right),
\end{aligned}
\end{equation}                                                                           
\noindent where

\begin{equation}
\begin{aligned}
\label{eq:imp}
\quad \mathcal{P}_{\mathcal{I}} = -\frac{1}{2} \frac{\left(\sum_{i \in \mathcal{I}} g_{i}\right)^{2}}{\sum_{i \in \mathcal{I}} h_{i}+\lambda}
\end{aligned}
\end{equation}
\noindent is the impurity of a node and $\mathcal{I}_{L}^{q},\mathcal{I}_{R}^{q}$ are the index sets associated with the $q$-th node's left and right child nodes after the split respectively.

\begin{table}[!t]
	\centering
	\caption{Summary of Notations}
	\label{tab:notation}
	\setlength{\tabcolsep}{2mm}{
		\begin{tabular}{llll}
			\toprule Notation & Description \\
			\hline
			$P_m$ & Participant $m$\\
			$n$ & Number of participants \\
			$N$ & Number of samples owned by each participant \\
			$J$ & Number of total features\\
			$J_{m}$ & Number of features owned by participant $m$ \\
			$\mathcal{D}^{N\times J_m}_m$ & Vertically partitioned dataset owned by participant $m$ \\
			$\mathbf{y}$ & Label set \\
			$D$ & Dimension of decision table \\
			$\mathcal{T}$ & Decision table model \\
			$T$ & Number of decision tables to be trained \\
			{$\llbracket{\mathbf{x}}\rrbracket$} & Secret-shared vector \\
			{$\langle{\mathbf{x}}\rangle_m$} & One share of a vector held by participant $m$ \\
			
			\hline
	\end{tabular}}
\end{table}
\subsection{Additive Secret Sharing}

In {\main}, we use $n$-out-of-$n$ additive secret sharing over $\mathbb{Z}_{2^{Q}}$, where $Q$ denotes the number of bits for value representation. 
%
%
In such secret sharing, a secret value $x\in \mathbb{Z}_{2^Q}$ is additively split into $n$ secret shares $\langle x\rangle_1 ,\langle x \rangle_2,\cdots,\langle x \rangle_n \in \mathbb{Z}_{2^{Q}}$ such that $\llbracket x \rrbracket=\langle x\rangle_1+\langle x \rangle_2+\cdots+\langle x \rangle_n$ mod$ \; 2^Q$.
The $n$ shares are held by $n$ parties respectively to be engaged in a secure computation.
For simplicity, we denote such additive secret sharing of $x$ by $\llbracket x \rrbracket$.
Below we introduce the basic operations related to additive secret sharing in the $n$-party setting.


$\bullet \mathsf{Sharing}$: To additively share a private value $x$ of party $P_l$, $P_l$ needs to generate $n-1$ random numbers $\{x_{m\neq l}\},m\in[1,n]$ over $\mathbb{Z}_{2^{Q}}$ and sends $x_{m\neq l}$ to $P_{m\neq l}$, respectively.
Then $P_l$ holds $ \langle x\rangle_l =(x- \sum_{p=1,p\neq l}^{n}x_p) \; $mod$ \; 2^Q$ and $P_{m\neq l}$ holds $\langle x \rangle_{m\neq l}=x_{m\neq l}$, respectively, as a share of $x$. For conciseness, the modulo operation will be henceforth omitted in the following protocols.

$\bullet \mathsf{Reconstruction}$: To reconstruct ($\mathsf{Rec(\cdot)}$) a shared value $\llbracket x \rrbracket$ on $P_l$, $P_{m\neq l}$ sends its share $\langle x\rangle_{m\neq l}$ to $P_l$ and $P_l$ computes $\sum_{p=1}^{n}\langle x\rangle_p$.

$\bullet \mathsf{Addition}$: For the two secret-shared values $\llbracket x \rrbracket$ and  $\llbracket y \rrbracket$ , to securely compute addition ($\llbracket z\rrbracket=\llbracket x\rrbracket+\llbracket y \rrbracket$), each participant $P_{m}$ locally computes $\langle z \rangle_m=\langle x \rangle_m+\langle y \rangle_m$. 
Similarly, to compute subtraction ($\llbracket z\rrbracket=\llbracket x\rrbracket-\llbracket y\rrbracket)$, each participant subtracts its local share of $y$ from that of $x$.


$\bullet \mathsf{Multiplication}$: To multiply a secret-shared value $\llbracket x \rrbracket$ with a constant $c$ ($\llbracket z \rrbracket=c\times{\llbracket x \rrbracket}$), each participant multiplies its local share of $x$ by $c$.
To multiply  two secret-shared values $\llbracket x \rrbracket,\llbracket y \rrbracket$ (denoted by $\llbracket z \rrbracket=\llbracket x \rrbracket \times \llbracket y\rrbracket$ where $z=xy$), the multiplication triple technique can be used \cite{beaver}. In an offline  phase, all parties obtain a secret-shared multiplication triple $(\llbracket a\rrbracket,\llbracket b\rrbracket,\llbracket c\rrbracket)$, where $a,b$ are uniformly random numbers in $\mathbb{Z}_{2^{Q}}$ and $c=ab$. 
The secret-shared triples are data-independent and can be prepared and distributed offline by an independent third-party \cite{Chameleon}, so hereafter we assume the triples are available for use in online secure computation among the parties.
The secret-shared multiplication proceeds as follows.
Each party $P_{m}$ locally computes $\langle e \rangle_m=\langle x \rangle_m-\langle a \rangle_m$ and $\langle f \rangle_m=\langle y \rangle_m-\langle b \rangle_m$. After that, the parties run $\mathsf{Rec}(\llbracket e\rrbracket),\mathsf{Rec}(\llbracket f\rrbracket)$. Next, $P_m$ computes $\langle z \rangle_m=j\times e\times f + f\times\langle a\rangle_m +e\times\langle b \rangle_m+\langle c \rangle_m$, where $j=1$ if $m=1$ and $j=0$ if $m\neq1$. 
Table \ref{tab:notation} summarizes the key notations in this paper.
%


\section{System Overview}
\label{sec:overview}
\subsection{System Architecture}

\begin{figure}[!t]
	\centering
	\includegraphics[scale=0.5]{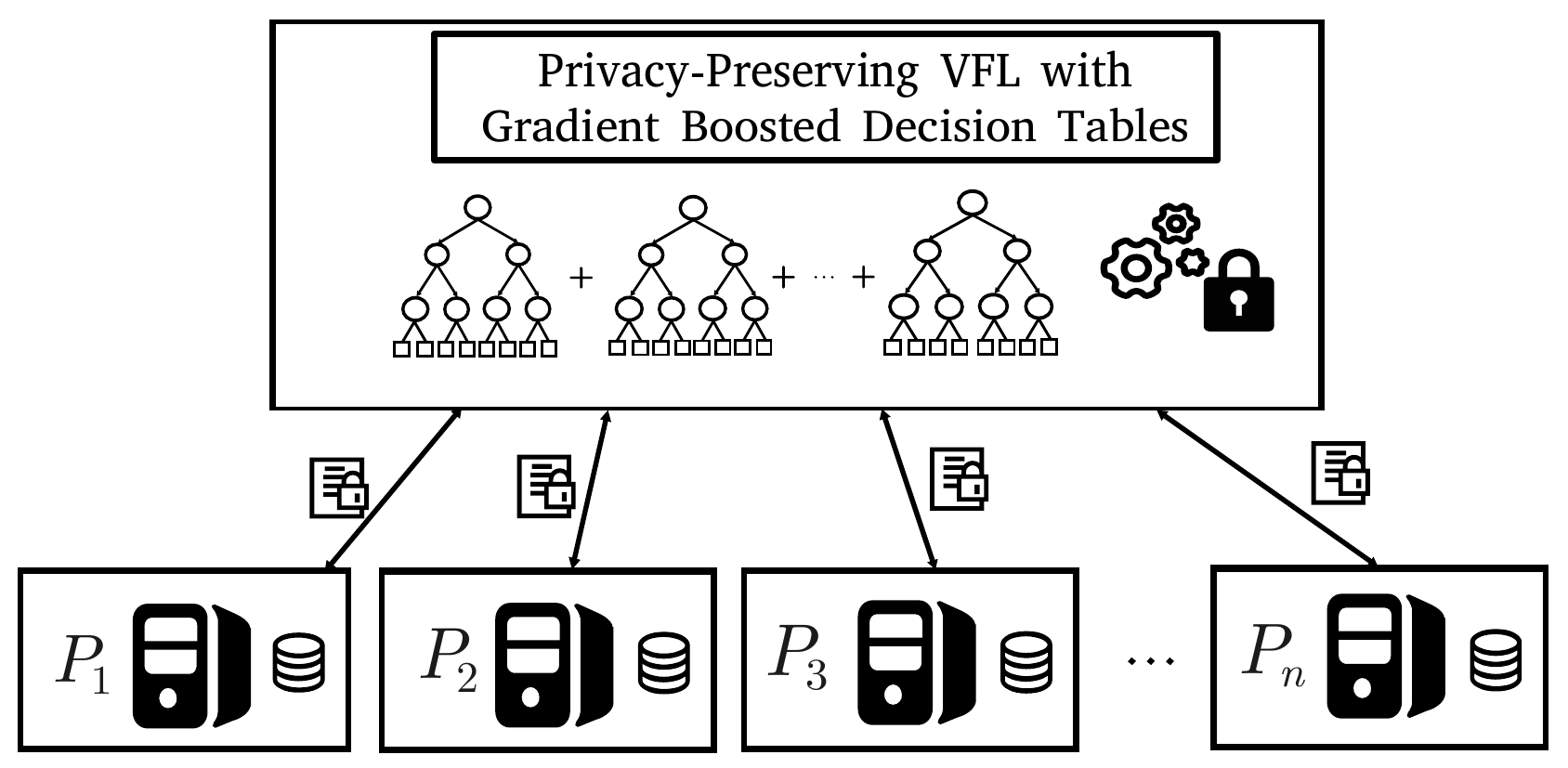}
	\caption{\revise{{\main}'s system architecture.}}
	\label{fig:arch}
\end{figure}

\revise{Fig. \ref{fig:arch} illustrates the system architecture of {\main}, which targets the vertical federated learning scenario. 
In {\main}, multiple participants (e.g., business organizations and institutions) want to collaboratively train gradient boosted decision tables over vertically partitioned data. Under such setting, a dataset consisting of $N$ samples (each is associated with a feature vector and a label) is vertically partitioned among $n$ participants $P_1,P_2,\cdots,P_n$.} 
Each participant $P_m$ holds its respective dataset $\mathcal{D}^{N\times J_m}_m=\{\mathbf{x}_i^{J_m}\}^{N-1}_{i=0}$, where $J_m$ denotes the number of features owned by $P_m$ and is subject to $\sum_{m=1}^{n}J_m=J$, $\mathbf{x}_i^{J_m}$ represents the $i$-th sample of $\mathcal{D}^{N\times J_m}_m$. 
Let $\mathbf{y}=\{y_i\}^{N-1}_{i=0}$ be the set of sample labels.
Following prior works on VFL \cite{cheng2021secureboost,tist,fang2021large}, we consider two roles for the participants: active participant (AP) and passive participant (PP). 
In particular, there is one AP that holds a local dataset as well as the label set $\mathbf{y}$; and the remaining participants are PPs, each only holding a local dataset. 
\revise{
For simplicity, in {\main}, we assume the participant $P_n$ is the AP. 
}


Throughout the secure training process in {\main}, each participant keeps its feature data locally.
The Boolean tests and output values of the decision tables are securely learned in {\main}, in such a manner that no participant knows the complete models. In particular, {\main} follows a setting similar to the works \cite{cheng2021secureboost,tist,fang2021large,tian2020federboost} under VFL, where each participant learns partial information of the learned models.
Specifically, in {\main}, all the participants know the split feature of each test in a decision table of the ensemble and who owns this feature, but only the participant owning this feature knows the split threshold of the test. 
Formally, for the learned test $F_d<t_d$ at the $d$-th level ($d\in[0,D-1]$) of the $t$-th decision table ($t\in[1,T]$) in the ensemble, the split feature $F_d$ is revealed to all participants but the threshold $t_d$ is only known by the participant owing the feature data corresponding to $F_d$. 
In addition, all the output values of leaf nodes are produced in secret-shared form among all participants. 


\subsection{Threat Model}

\revise{{\main} is designed under the semi-honest adversary model, as is common in state-of-the-art security designs on vertical federated learning \cite{fang2021large,wu2020}. 
	Specifically, in {\main}, each participant is assumed to faithfully follow the protocol specification but may try to deduce other participants' private information from the messages they receive. 
	It is noted that though we consider two roles AP and PP for the participants, no additional trust is assumed regarding the AP.
	The semi-honest adversary model should be reasonable in practice because VFL aims at breaking down the data silos between business organizations, where the behavior of each organization is strictly enforced by privacy regulations\cite{GDPR}. 
	We also consider that a static adversary may corrupt a subset of $\tau$ participants ($\tau\leq n-1$). That is, a static adversary may choose a subset of the participants to corrupt before the VFL procedure and the chosen participants remain corrupted during the VFL procedure.}

\revise{
	Under the above threat model, {\main} aims to guarantee that a semi-honest participant individually cannot learn any other participant's local data and learned partial model (tests and output values of each decision table in the ensemble) throughout the VFL procedure. 
	In case of collusion among a subset of the participants, {\main} strives to still ensure that the honest participants' private information is protected against the corrupted participants. 
 	Like prior works \cite{fang2021large,popets,cheng2021secureboost}, {\main} does not hide the data-independent generic parameters, such as the dimension $D$ and the number of decision tables $T$.
 	}

\section{The Design of {\main}}
\label{sec:design}
\subsection{Overview}

%
We provide in Algorithm \ref{alg:secure-framework} an overview of the secure training framework in {\main}, which inputs the vertically partitioned datasets $\{\mathcal{D}^{N\times J_m}_m\}^n_{m=1}$ and the label set $\mathbf{y}$ from the participants, and outputs an ensemble $\mathcal{E}$ of distributed decision tables among the participants.  
At the beginning, the secret-shared inference result $\llbracket{\mathbf {\hat{y}}^{(0)}}\rrbracket$ is initialized as the secret sharing $ \llbracket\mathbf{0}_{N}\rrbracket $ ($N$ denotes the length of the secret-shared vector), and the AP distributes the secret shares of its label set to other participants. 
After that, $T$ distributed decision tables are securely built sequentially in $T$ rounds (lines \ref{2-line5}-\ref{2-line9}).

We develop a secure decision table learning algorithm $\mathsf{SecTable}$ to support the secure training of a single (distributed) decision table in each round.
$\mathsf{SecTable}$ consists of several secure components, including (i) secure node splitting $\mathsf{SecSplit}$, (ii) secure Sigmoid evaluation $\mathsf{SecSigmoid}$, and (iii) secure discretization $\mathsf{SecDisc}$. 
The secure node splitting component $\mathsf{SecSplit}$ (Section \ref{secsplit}) is to securely split the nodes at a certain level and partition the index sets associated with these nodes without revealing the partitioned index sets. 
The secure Sigmoid evaluation component $\mathsf{SecSigmoid}$ (Section \ref{secsigmoid}) inputs a secret-shared value and calculates the Sigmoid function in the secret sharing domain. 
The secure discretization component $\mathsf{SecDisc}$ (Section \ref{secdisc}) is to securely rearrange the secret-shared gradients according to the local permutations owned by each participant and then group them into buckets. 
Through the synergy of these components, $\mathsf{SecTable}$ allows the participants to securely train a distributed decision table in each round, for which we will give the details in Section \ref{subsec:sectable}.

After securely training a distributed decision table in a certain round, secure inference needs to be conducted, of which the result will be added to previous inference results (line \ref{2-line8}) for use in $\mathsf{SecTable}$ in the next round.
To this end, we develop a secure distributed decision table inference protocol $\mathsf{SecInfer}$ (Section \ref{SecInfer}), which inputs each participant's local data and partial model to produce secret-shared inference results without leaking their data and partial model. 
It is worth noting that $\mathsf{SecInfer}$ can also be used to support secure inference for new data after the completion of the whole training process.

\begin{algorithm}[!t]
	\caption{Overview of Our Secure Training Framework} 
	\label{alg:secure-framework}
	\begin{algorithmic}[1]
		\Require
		$P_1,P_2,\cdots,P_n$ hold local datasets $\{\mathcal{D}^{N\times J_m}_m\}^n_{m=1}$ and one AP holds the label set $\mathbf y$.
		\Ensure
		$P_1,P_2,\cdots,P_n$ obtain an ensemble $\mathcal{E}$ of $T$ distributed decision tables $\{\mathcal{T}_t\}^T_{t=1}$, each containing $D$ tests and $2^D$ secret-shared output values.
		
		\State Initialize ensemble $\mathcal{E} = \{\}$.
		\State  $\left\llbracket{\mathbf {\hat{y}}^{(0)}}\right\rrbracket \leftarrow \llbracket\mathbf{0}_{N}\rrbracket \:$.
		\State AP secret-shares $\mathbf{{y}}$ to other participants to produce $\left\llbracket{\mathbf{{y}}}\right\rrbracket$.
		
		\For {$t\in[1,T]$} \label{2-line5}
		\State $\mathcal{T}_t\leftarrow $SecTable($\{\{\mathcal{D}^{N\times J_m}_m\}^n_{m=1},\mathbf{\llbracket y\rrbracket},\llbracket \mathbf{\hat{y}}^{(t-1)} \rrbracket$) \label{2-line6}.
		\State $\left\llbracket{\mathbf {\hat{y}}^{(t)}}\right\rrbracket \leftarrow $SecInfer$(\mathcal{T}_t$, $\{\{\mathcal{D}^{N\times J_m}_m\}^n_{m=1})$ \label{2-line7}.
		\State $\left\llbracket{\mathbf{\hat{y}}^{(t)}}\right\rrbracket \leftarrow \left\llbracket{\mathbf{\hat{y}}^{(t-1)}}\right\rrbracket  + \left\llbracket{\mathbf{\hat{y}}^{(t)}}\right\rrbracket$ \label{2-line8}.
		\EndFor \label{2-line9}
		\State $\mathcal{E}.append(\mathcal{T}_t)$.
		\State Output the ensemble $\mathcal{E}$ held by $P_1,P_2,\cdots,P_n$. 
	\end{algorithmic}

\end{algorithm}

\subsection{Secure Distributed Decision Table Training}


\subsubsection{Secure Node Splitting}
\label{secsplit}

An oblivious tree grows to a new level by splitting each node at the current level into two child nodes. 
In plaintext centralized decision table training (Algorithm \ref{algo:plaintext-oblivious-tree}), node splitting is performed by partitioning the samples associated with the node to be split. 
However, in VFL, the partitioning of samples must be revealed to all participants because the training dataset is vertically partitioned and all participants hold the same samples. For instance, given a test ``$Height<180$'', the participant owning feature data of ``$Height$'' needs to tell other participants which samples are less than $180$ and which samples are greater than $180$, which will leak each sample's range of ``$Height$'' and raise critical privacy concerns.

To avoid this leakage, we design a secure node splitting component $\mathsf{SecSplit}$. 
Inspired by existing works\cite{wu2020,fang2021large,tist}, we utilize indicator vectors to conduct secure node splitting for each level of the decision table. 
At a high level, {\main} associates the $k$-th node ($k\in[0,2^d-1]$) at the $d$-th level ($d\in[0,D-1]$) of the oblivious decision tree with a first-order gradient vector $\llbracket \mathbf{g}^{k,d}\rrbracket$ and a second-order gradient vector $\llbracket \mathbf{h}^{k,d}\rrbracket$, each containing $N$ elements that are secret-shared among all participants $P_1,P_2,\cdots,P_n$. 
If the $i$-th sample is partitioned into this node, the $i$-th element in $\llbracket \mathbf{g}^{k,d}\rrbracket$ and $\llbracket \mathbf{h}^{k,d}\rrbracket$ will be set as the $i$-th sample's first and second-order gradients, respectively, otherwise  the $i$-th element will be set as $ \llbracket0\rrbracket$. 
%
%

Algorithm \ref{alg::secsplit} gives the procedure of secure node splitting. 
Firstly, participant $P_l$ who owns the optimal test $F_d<t_d$ at the $d$-th level locally generates left indicator vector $\mathbf{v}_l$ and right indicator vector $\mathbf{v}_r$ and then distributes their secret sharings (denoted by $ \llbracket \mathbf{v}_l \rrbracket$ and $ \llbracket \mathbf{v}_r \rrbracket$) to other participants (i.e., lines \ref{secplit-1}-\ref{secplit-8}  in Algorithm \ref{alg::secsplit}).
Upon receiving $ \llbracket \mathbf{v}_l \rrbracket$ and $ \llbracket \mathbf{v}_r \rrbracket$, for the $k$-th node at this level, $P_1,P_2,\cdots,P_n$ update the first-order and second-order gradient vector of the $k$-th node's left and right child nodes. The update is achieved with secure element-wise multiplication between secret-shared indicator vectors and gradient vectors (i.e., lines  \ref{secplit-9}-\ref{secplit-11} in Algorithm \ref{alg::secsplit}). 
In this way, the index set processed by each node is hidden and the number of the samples processed by each node remains constant as $N$, which means an adversary cannot deduce any information from node splitting.

\begin{algorithm}[!t]
	\caption{Secure Node Splitting ($\mathsf{SecSplit}$)} 
	\label{alg::secsplit}
	\begin{algorithmic}[1]
		\Require
		$P_1,P_2,\cdots,P_n$ hold local datasets $\{\mathcal{D}^{N\times J_m}_m\}^n_{m=1}$, the secret-shared first and second-order gradient vectors $\{\llbracket \mathbf{g}^{k,d}\rrbracket\}^{2^d-1}_{k=0}$, $\{\llbracket \mathbf{h}^{k,d}\rrbracket\}^{2^d-1}_{k=0}$ associated with $2^d$ nodes at the $d$-th level;
		$P_{l}$ holds the optimal test $F_d<t_d$ at the $d$-th level.
		\Ensure $P_1,P_2,\cdots,P_n$ obtain the secret-shared gradient vectors $\{\llbracket \mathbf{g}^{k,d+1}\rrbracket\}^{2^{d+1}-1}_{k=0}$ and $\{\llbracket \mathbf{h}^{k,d+1}\rrbracket\}^{2^{d+1}-1}_{k=0}$ associated with the $2^{d+1}$ split nodes.

		\leftline{// \underline{$P_l$ performs:}}
		\State $\mathbf{v}_l\leftarrow \mathbf{0}_{N}$\label{secplit-1}.
		\For {$i\in[0,N-1]$} 
		\If{$\mathbf{x}_i^{J_l}(F_d)<t_d$}
		\State $\mathbf{v}_{li}\leftarrow 1$ // Set the $i$-th element of $\mathbf{v}_l$ as $1$.
		\EndIf 
		\EndFor
		\State $\mathbf{v}_r\leftarrow \mathbf{1}_{N}-\mathbf{v}_l$.
		\State $P_l$ secret-shares $\mathbf{v}_l $ and $\mathbf{v}_r$ to other participants. \label{secplit-8}
		
		\leftline{// \underline{$P_1,P_2,\cdots,P_n$ perform:}}
		\For{$k\in[0,2^d-1]$} \label{secplit-9}

		\State $\llbracket \mathbf{g}^{2k,d+1} \rrbracket \leftarrow \llbracket \mathbf{v}_l \rrbracket\times \llbracket \mathbf{g}^{k,d} \rrbracket$.
		\State $\llbracket \mathbf{h}^{2k,d+1} \rrbracket \leftarrow \llbracket \mathbf{v}_l \rrbracket\times \llbracket \mathbf{h}^{k,d} \rrbracket$.
		
		\State $\llbracket \mathbf{g}^{2k+1,d+1} \rrbracket \leftarrow \llbracket \mathbf{v}_r \rrbracket\times \llbracket \mathbf{g}^{k,d} \rrbracket$.
		\State $\llbracket \mathbf{h}^{2k+1,d+1} \rrbracket \leftarrow \llbracket \mathbf{v}_r \rrbracket\times \llbracket \mathbf{h}^{k,d} \rrbracket$.

		\EndFor \label{secplit-11}
		\State Output the secret-shared gradient vectors $\{\llbracket \mathbf{g}^{k,d+1}\rrbracket\}^{2^{d+1}-1}_{k=0}$ and $\{\llbracket \mathbf{h}^{k,d+1}\rrbracket\}^{2^{d+1}-1}_{k=0}$ associated with nodes at the $(d+1)$-th level.
		
	\end{algorithmic}
\end{algorithm}

\subsubsection{Secure Sigmoid Evaluation}
\label{secsigmoid}

\revise{
There are mainly two challenges in securely calculating the Sigmoid function in the secret sharing domain. Firstly, how to compute the division $\llbracket \frac{x}{y}\rrbracket$ given two secret-shared values $\llbracket x\rrbracket$ and $\llbracket y\rrbracket$? Secondly, how to compute the exponentiation function $ \llbracket e^{x} \rrbracket$ given a secret-shared value $\llbracket x\rrbracket$?
Next, we introduce how {\main} tackles the two challenges so as to allow the participants to securely calculate the Sigmoid function in the secret sharing domain.}
For the first challenge, we introduce a secure division component $\mathsf{SDiv}$ by transforming the division calculation into a numerical optimization problem.
Specifically, we note that the core obstacle of calculating $\llbracket \frac{x}{y}\rrbracket$ given $\llbracket x\rrbracket$ and $\llbracket y\rrbracket$ is to calculate the secret-shared reciprocal $\llbracket \frac{1}{y}\rrbracket$.
Therefore, we first approximate $\frac{1}{y}$ by the iterative Newton-Raphson algorithm \cite{verbeke1995newton}, following previous works \cite{cryptgpu,crypten}: $z_{i} \leftarrow 2 z_{i-1}-y z_{i-1}^{2}$, which will converge to $z_n\approx\frac{1}{y}$. 
In {\main}, we fix the initialization $z_0=1/Y$, where $Y$ is a sufficiently large value.
Note that the approximation consists of basic subtraction and multiplication operations which are naturally supported in the secret sharing domain, given $\llbracket y\rrbracket$, the secret-shared reciprocal $\llbracket \frac{1}{y}\rrbracket$ can be securely calculated. 
After securely calculating the  reciprocal $\llbracket \frac{1}{y}\rrbracket$, {\main} multiplies $\llbracket x\rrbracket$ by $\llbracket \frac{1}{y}\rrbracket$ to obtain $\llbracket \frac{x}{y}\rrbracket$, i.e., $\llbracket \frac{x}{y}\rrbracket=\llbracket x\rrbracket\cdot\llbracket \frac{1}{y}\rrbracket$.

For the second challenge, i.e., computing the exponentiation function $ \llbracket e^{x} \rrbracket$ given a secret-shared value $\llbracket x\rrbracket$, we approximate $ e^{x} $ by limit characterization, inspired by \cite{cryptgpu}: $e^{x}\approx(1+\frac{x}{2^n})^{2^n}$, which provides a good approximation of $ e^{x} $.
Note that since the approximation consists of basic addition and multiplication operations which are naturally supported in the secret sharing domain, given $\llbracket x\rrbracket$, the secret-shared exponentiation $\llbracket e^{x}\rrbracket$ can be securely calculated.
However, we note that the approximation method requires $2^n$ chain multiplications, and thus requires $2^n$ rounds of online communication.
The approximation method is inefficient in practice since the communication complexity grows exponentially.
Therefore, {\main} further reduces the exponential communication complexity to linear communication complexity. 
More specifically, we note that the computation in the approximation can be regarded as $ a^{2^n}$ where $ a=1+\frac{x}{2^n}$.
Therefore, given $\llbracket a\rrbracket$, {\main} first securely calculates  $\llbracket a^2\rrbracket$, which only requires one round of communication.
After that, {\main} regards the output as $\llbracket y\rrbracket=\llbracket a^2\rrbracket$ followed by securely calculating $\llbracket y^2\rrbracket$, which also only requires one round of communication.
Therefore, in this way, we can securely calculate $\llbracket a^{2^n}\rrbracket$ in $\log 2^{n}=n$ rounds instead of $2^{n}$ rounds. 
\revise{
Clearly, there is a trade-off between accuracy and efficiency in approximating $e^x$ with $(1+\frac{x}{2^n})^{2^n}$ for computation in the secret sharing domain. 
In principle increasing the value of $n$ would lead to a more accurate approximation of the Sigmoid function. However, this also leads to increased computation and communication costs. 
Yet, as will be shown by our experiments, a small value of $n$ (in our case, we set $n=2$) suffices to enable {\main} to achieve the accuracy comparable to plaintext centralized learning.
}

\noindent \revise{\textbf{Remark.}} \revise{In the literature, there exist some methods for approximating the Sigmoid function so as to support secure Sigmoid evaluation, including Taylor expansion\cite{taylor}, piece-wise approximation \cite{mohassel2017secureml}, and function approximation like $f(x)=\frac{0.5x}{1+|x|}+0.5$ \cite{fang2021large}. 
%
%
For the Taylor expansion method, it requires a small input parameter (very close to 0), which is hard to satisfy in machine learning.
The piece-wise approximation method has no such requirement but suffers from notable accuracy loss\cite{fang2021large}.
The work that is most closely related to ours is due to Fang \textit{et al.} \cite{fang2021large}, who apply another function approximation method to approximate the Sigmoid function, i.e., $f(x)=\frac{0.5x}{1+|x|}+0.5$. 
However, their method still experienced non-trivial loss in accuracy in their securely trained XGBoost model. 
As reported in their experiments, the Area Under the ROC Curve (AUC) value would go up to \emph{0.84463} from 0.82945 if they replace the secure Sigmoid approximation with \emph{plaintext} Sigmoid computation. 
In contrast, as will be shown by the experiments in Section \ref{subsec:utility_evaluation}, our proposed $\mathsf{SecSigmoid}$ can achieve AUC values that are highly close to those obtained using plaintext centralized learning (e.g., the gap can be as small as 0.0005).
}

\subsubsection{Secure Discretization}
\label{secdisc}
Discretization, also called bucketing, is a commonly used grouping method in large-scale machine learning \cite{xgboost,lightgbm}.
Specifically, discretization groups the samples into a small number of buckets so as to allow the model training to scale on larger datasets.
Let $B$ denote the number of buckets in discretization, where $B\ll N$ and $N$ is the number of samples. 
In gradient boosting, gradients are grouped into buckets and the sum of gradients in each bucket is calculated in the training stage\cite{xgboost,lightgbm}. 
Typically, for each feature, the gradients are first permuted by a permutation $\pi$, which is obtained by sorting the values of this feature. 
Then the permuted gradients are partitioned into $B$ buckets. 
Obviously, the cost of training on $B$ buckets instead of $N$ samples can be greatly reduced. 

\revise{
However, discretization is non-trivial in privacy-preserving machine learning. 
In existing MPC-based works\cite{adams2022privacy,deforth2021xorboost,popets}, the training data is secret-shared among a fixed set of computing servers, and sorting the secret-shared training data for discretization requires a large number of secure comparison operations, which is expensive in the secret sharing domain.
In {\main}, the training data is vertically partitioned, and thus the sorting process can be achieved locally to reduce the overhead.
However, it is still difficult to permute the secret-shared gradient vector by a permutation $\pi$ held by a participant $P_l$ without revealing $\pi$ to other participants. 
}
%
%

\begin{algorithm}[!t]
	\caption{Secure Permutation ($\mathsf{SecPerm}$)} 
	\label{alg:SecPerm}
	\begin{algorithmic}[1]
		\Require
		$P_1,P_2,\cdots,P_n$ hold the secret-shared vector $\llbracket \mathbf{x} \rrbracket$; $P_l$ inputs the permutation $\pi$.
		\Ensure
		$P_1,P_2,\cdots,P_n$ obtain the permuted secret-shared vector $\llbracket \mathbf{u} \rrbracket$ subject to $\mathbf{u}=\pi(\mathbf{x})$.
		
		\leftline{// \underline{Initialization:}} 
		\State $P_1,P_2,\cdots,P_n$ hold in advance the secret shares of $\pi_p(\mathbf{r})$ and $\mathbf{r}$, and $P_l$ additionally holds $\pi_p$. \label{3-line2}
		
		\leftline{// \underline{Online computation:}}
		
		\leftline{ \textbf{Round 1:}}
		\State $P_l$ generates $\pi_s$ subject to $\pi(\cdot)=\pi_s[\pi_p(\cdot)]$, and then sends $\pi_s$ to other participants. \label{3-line5}
		\State Each participant $P_{m}$ locally calculates $\langle \mathbf{x}\rangle_m -\langle\mathbf{r} \rangle_m  $.
		
		\leftline{ \textbf{Round 2:}}
		\State $P_1,P_2,\cdots,P_n$ run $\mathsf{Rec}(\llbracket \mathbf{x-r} \rrbracket)$ to reveal $\mathbf{x-r}$ to $P_l$.\label{3-line7}
		\State $P_l$: $\langle \mathbf{u} \rangle_l \leftarrow \pi(\mathbf{x-r}) + \pi_s(\langle \pi_p(\mathbf{r}) \rangle_l)$. \label{3-line8}
		\State $P_{m\neq l}$: $\langle \mathbf{u} \rangle_{m\neq l} \leftarrow \pi_s(\langle \pi_p(\mathbf{r}) \rangle_{m\neq l})$. \label{3-line9}
		\State Output the secret-shared vector $\llbracket \mathbf{u} \rrbracket$ held by $P_1,P_2,\cdots,P_n$.

	\end{algorithmic}
\end{algorithm}

\begin{algorithm}[!t]
	\caption{Secure Discretization ($\mathsf{SecDisc}$)} 
	\label{alg::secdisc}
	\begin{algorithmic}[1]
		\Require
		$P_1,P_2,\cdots,P_n$ hold the secret-shared first and second-order gradient vectors $\llbracket{\mathbf{g}}\rrbracket$ and $\llbracket{\mathbf{h}}\rrbracket$; $P_l$ holds the permutation $\pi$.
		\Ensure
		$P_1,P_2,\cdots,P_n$ obtain two secret-shared vectors $\llbracket\boldsymbol{\alpha}\rrbracket$ and $\llbracket\boldsymbol{\beta}\rrbracket$ containing $B$ grouped first and second-order gradients, respectively.
		
		\State $\llbracket{\mathbf{g'}}\rrbracket\leftarrow$ $\mathsf{SecPerm}$($\pi,\llbracket{\mathbf{g}}\rrbracket$).
		\State $\llbracket{\mathbf{h'}}\rrbracket\leftarrow$ $\mathsf{SecPerm}$($\pi,\llbracket{\mathbf{h}}\rrbracket$).
		
		\State $\llbracket\boldsymbol{\alpha}\rrbracket \leftarrow \llbracket\mathbf{0}_{B}\rrbracket,\llbracket\boldsymbol{\beta}\rrbracket \leftarrow \llbracket\mathbf{0}_{B}\rrbracket$.
		\State $M\leftarrow N/B$ \label{disc-line6}.
		
		\For{$b\in[0,B-1]$} 
		\State $\llbracket \boldsymbol{\alpha}_{b}\rrbracket \leftarrow \sum_{i=b\times M}^{(b+1)\times M-1}\llbracket \mathbf{g}'_{i}\rrbracket$.
		\State $\llbracket \boldsymbol{\beta}_{b}\rrbracket \leftarrow \sum_{i=b\times M}^{(b+1)\times M-1}\llbracket \mathbf{h}'_{i}\rrbracket$.
		\EndFor \label{disc-line9}
		\State Output the secret-shared vectors $\llbracket\boldsymbol{\alpha}\rrbracket$ and $\llbracket\boldsymbol{\beta}\rrbracket$ held by $P_1,P_2,\cdots,P_n$.
		
	\end{algorithmic}
\end{algorithm}

\revise{
To tackle the challenge, we propose a secure permutation algorithm $\mathsf{SecPerm}$ (shown in Algorithm \ref{alg:SecPerm}), which stems from the correlated randomness (CR) scheme in \cite{fang2021large}.
Our tailored design $\mathsf{SecPerm}$ enables our secure discretization component $\mathsf{SecDisc}$ to outperform that in  \cite{fang2021large}  in supporting an arbitrary number of participants. 
Fang \textit{et al.} \cite{fang2021large} design two MPC-based secure discretization methods.
Specifically, they first propose a basic discretization method based on multiplications between secret-shared lagre-scale matrices.
Then they obtain significant speedup over the basic method by utilizing CR to efficiently permute secret-shared gradients, and then group them. 
However, both the basic and improved methods in \cite{fang2021large} only work under the two-party setting in VFL.
The work \cite{tist} is the first MPC-based work supporting more than two participants in VFL with GBDT, but it simply follows the basic discretization method in \cite{fang2021large}.
In contrast, {\main} tailors the improved discretization method from \cite{fang2021large} to support an arbitrary number of participants, which is more efficient than the straightforward secret-shared lagre-scale matrix multiplication-based method from \cite{tist}. 
}

As shown in Algorithm \ref{alg:SecPerm}, at the beginning of $\mathsf{SecPerm}$,  participants $P_1,P_2,\cdots,P_n$ hold a secret-shared vector $\llbracket \mathbf{x} \rrbracket$ and $P_l$ holds a permutation $\pi$.
At the end of $\mathsf{SecPerm}$, participants $P_1,P_2,\cdots,P_n$ hold a secret-shared vector $\llbracket \mathbf{u} \rrbracket$ where $\mathbf{u}=\pi(\mathbf{x})$.
$\mathsf{SecPerm}$ guarantees that except for $P_l$, other participants cannot know the  permutation $\pi$.
In the initialization of Algorithm \ref{alg:SecPerm}, all participants hold in advance the secret shares of $\pi_p(\mathbf{r})$ and $\mathbf{r}$, and $P_l$ additionally holds $\pi_p$. 
After the initialization, all participants collaboratively permute $\llbracket \mathbf{x} \rrbracket$ in $2$ rounds. 
In the first round, $P_l$ generates the permutation $\pi_s$, which subjects to $\pi(\cdot)=\pi_s[\pi_p(\cdot)]$, and then $P_l$ sends  $\pi_s$  to all other participants. 
After that, each participant $P_{m}$ locally calculates $\langle \mathbf{x}\rangle_m -\langle\mathbf{r} \rangle_m  $. 
In the second round, $ \mathbf{x-r}$ is revealed to $P_l$.
Finally, the participants output $\llbracket\mathbf{u}\rrbracket= \llbracket \pi(\mathbf{x}) \rrbracket$ ( i.e., lines \ref{3-line8}-\ref{3-line9} in Algorithm \ref{alg:SecPerm}). 

The correctness analysis of $\mathsf{SecPerm}$ is as follows: 
\begin{equation}
\begin{aligned}
\mathbf{u} &= \langle \mathbf{u} \rangle _1 + \langle \mathbf{u} \rangle _2 + \cdots+\langle \mathbf{u} \rangle _l+\cdots+\langle \mathbf{u} \rangle _n\\
&= \pi_s(\langle \pi_p(\mathbf{r}) \rangle_1)+\pi_s(\langle \pi_p(\mathbf{r}) \rangle_2)+\cdots\\
&\quad+ \pi(\mathbf{x-r}) + \pi_s(\langle \pi_p(\mathbf{r}) \rangle_l)+\cdots+\pi_s(\langle \pi_p(\mathbf{r}) \rangle_n)\\
&=\pi(\mathbf{x-r})+\pi_s( \pi_p(\mathbf{r}) )\\
&=\pi(\mathbf{x-r})+\pi(\mathbf{r})\\
&=\pi(\mathbf{x}).
\end{aligned}
\nonumber
\end{equation}

Then, we introduce how {\main} securely realizes discretization protocol $\mathsf{SecDisc}$ based on $\mathsf{SecPerm}$.
At a high level, $\mathsf{SecDisc}$ first uses $\mathsf{SecPerm}$ to securely permute the secret-shared first and second-order gradients of $P_1,P_2,\cdots,P_n$  with a permutation $\pi$ held by $P_l $, and then partitions the gradients into $B$ buckets. 
Algorithm \ref{alg::secdisc} describes the details of our secure discretization protocol.

At the beginning, the secret-shared first and second-order gradient vectors $\llbracket{\mathbf{g}}\rrbracket$ and $\llbracket{\mathbf{h}}\rrbracket$ are securely permuted by $\mathsf{SecPerm}$, which outputs $\llbracket{\mathbf{g'}}\rrbracket$ and $\llbracket{\mathbf{h'}}\rrbracket$.
After that, $P_1,P_2,\cdots,P_n$ first initialize two secret-shared vectors $\llbracket\boldsymbol{\alpha}\rrbracket$ and $\llbracket\boldsymbol{\beta}\rrbracket$ of length $B$ to store the grouped first and second-order gradients, respectively. 
$\llbracket\boldsymbol{\alpha}\rrbracket$ and $\llbracket\boldsymbol{\beta}\rrbracket$ can be locally initialized as $\mathbf{0}_{B}$.
After that, for $b\in[0,B-1]$, the $b$-th bucket's secret-shared grouped first and second-order gradients $	\llbracket \boldsymbol{\alpha}_{b}\rrbracket $ and $\llbracket \boldsymbol{\beta}_{b}\rrbracket$ are calculated as follows: 
\begin{equation}\notag
\llbracket \boldsymbol{\alpha}_{b}\rrbracket = \sum_{i=b\times M}^{(b+1)\times M-1}\llbracket \mathbf{g}'_{i}\rrbracket; ~~~~
\llbracket \boldsymbol{\beta}_{b}\rrbracket = \sum_{i=b\times M}^{(b+1)\times M-1}\llbracket \mathbf{h}'_{i}\rrbracket,\notag
\end{equation}
where $M=N/B$ is the number of gradients in a bucket. 
For conciseness, we assume that $N$ can divide $B$ evenly.
As introduced in Section \ref{secsplit}, since the invalid gradients are set as $\llbracket 0\rrbracket$, the sum of the secret-shared gradients in a bucket is equal to that in the plaintext.


\subsubsection{Secure Decision Table Training Algorithm}
\label{subsec:sectable}

\begin{algorithm}[!t]
	\caption{Secure Training of a Distributed Decision Table ($\mathsf{SecTable}$)} 
	\label{alg::sectable}
	\begin{algorithmic}[1]
		
		\Require
		$P_1,P_2,\cdots,P_n$ hold local datasets $\{\mathcal{D}^{N\times J_m}_m\}^n_{m=1}$, secret-shared label set $\llbracket \mathbf{y} \rrbracket$ and aggregated inference results of the current models $\llbracket \mathbf{\hat{y}} \rrbracket$. 
		
		\Ensure
		$P_1,P_2,\cdots,P_n$ obtain a distributed decision table $\mathcal{T}$ with $D$ tests and a secret-shared vector $\llbracket \mathbf{w}\rrbracket$ of $2^D$  output values.

		\leftline{\space\space\space\space\space\space// \underline{Initialization:}}
		\If {the problem is regression} \label{4-1}
		\State  $\llbracket \mathbf{g} \rrbracket \leftarrow \llbracket\hat{\mathbf{y}} \rrbracket - \llbracket{\mathbf{y}} \rrbracket,\llbracket \mathbf{h} \rrbracket \leftarrow \llbracket \mathbf{1} \rrbracket$.
		\Else
		\State $\llbracket \mathbf{p} \rrbracket \leftarrow \mathsf{SecSigmoid}(\llbracket \mathbf{\hat{y}} \rrbracket)$.
		\State $\llbracket \mathbf{g} \rrbracket \leftarrow \llbracket{\mathbf{p}} \rrbracket - \llbracket{\mathbf{y}} \rrbracket,\llbracket\mathbf{h} \rrbracket \leftarrow\llbracket \mathbf{p} \rrbracket\times(\llbracket \mathbf{1}\rrbracket-\llbracket\mathbf{p} \rrbracket)$.
		\EndIf \label{4-6}
		
		\State $\llbracket \mathbf{w}\rrbracket \leftarrow \llbracket \mathbf{0}_{2^D}\rrbracket$.
		
		\leftline{\space\space\space\space\space\space// \underline{Training:}}
		\For {$d\in[0,D-1]$}	\label{4-line2}
		\State $\llbracket \boldsymbol{\sigma} \rrbracket \leftarrow \llbracket \mathbf{0}_{J} \rrbracket$, $ \boldsymbol{\gamma}  \leftarrow  \mathbf{0}_{J} $\label{4-line9}. 
		\For {$j\in[0,J-1]$ \textbf{in parallel}}	\label{4-line4}
		\State Participant who holds $\mathcal{D}(j)$ sorts $\mathcal{D}(j)$ to produce
		
		the permutation $\pi_j$. 
		\State $\llbracket \boldsymbol{\delta} \rrbracket \leftarrow \llbracket \mathbf{0}_{B-1} \rrbracket$.
		\For{$k\in [0,2^d-1]$}
		\State \leftline{$\llbracket\boldsymbol{\alpha}^{k,d}\rrbracket,\llbracket\boldsymbol{\beta}^{k,d}\rrbracket \leftarrow \mathsf{SecDisc}(\pi_j,\llbracket{\mathbf g^{k,d}}\rrbracket,\llbracket{\mathbf h^{k,d}}\rrbracket)$.}
		\State $\llbracket {G} \rrbracket,\llbracket {H} \rrbracket \leftarrow \sum_{i=0}^{B-1}\llbracket\boldsymbol{\alpha}^{k,d}_i\rrbracket,\sum_{i=0}^{B-1}\llbracket\boldsymbol{\beta}^{k,d}_i\rrbracket $.
		
		\For{$c\in[0,B-2]$}
		\State $\llbracket {G_{l}} \rrbracket,\llbracket{H_{l}}\rrbracket \leftarrow \sum_{i=0}^{c}\llbracket\boldsymbol{\alpha}^{k,d}_i\rrbracket,   \sum_{i=0}^{c}\llbracket\boldsymbol{\beta}^{k,d}_i\rrbracket$.
		\State \leftline{$\llbracket {G_{r}} \rrbracket,\llbracket{H_{r}}\rrbracket \leftarrow \llbracket {G} \rrbracket - \llbracket {G}_l \rrbracket, \llbracket {H} \rrbracket - \llbracket {H}_l \rrbracket$.}
		\State $\llbracket \boldsymbol{\delta}_{c}\rrbracket\leftarrow \llbracket \boldsymbol{\delta}_{c}\rrbracket-\frac{1}{2} \frac{\llbracket {G_{l}} \rrbracket^{2}}{\llbracket {H_{l}} \rrbracket+\llbracket\lambda\rrbracket}-\frac{1}{2} \frac{\llbracket {G_{r}} \rrbracket^{2}}{\llbracket {H_{r}} \rrbracket+\llbracket\lambda\rrbracket}$.
		
		\EndFor

		\EndFor
		\State $q \leftarrow \mathsf{SecArgmin}(\llbracket \boldsymbol{\delta} \rrbracket)$.
		\State $\llbracket \boldsymbol{\sigma}_{j} \rrbracket \leftarrow \llbracket \boldsymbol{\delta}_{q} \rrbracket$. \label{4-line23}
		\State $ \boldsymbol{\gamma}_{j}  \leftarrow q$. \label{4-line24}
		
		\EndFor \label{4-line16}
		\State $F_d \leftarrow \mathsf{SecArgmin}(\llbracket \boldsymbol{\sigma} \rrbracket)$ \label{4-line26} // Optimal split feature.
		\State $q_d \leftarrow \boldsymbol{\gamma}_{F_d} $  // Optimal bucket ID.\label{4-line27}
		\State $t_d \leftarrow \{\pi_{F_d}[\mathcal{D}(F_d)]\}_{(q_d+1)\times(N/B)}$  \label{4-line28}
		\State $\mathsf{SecSplit}$($\{\mathcal{D}^{N\times J_m}_m\}^n_{m=1},\{\llbracket \mathbf{g}^{k,d}\rrbracket\}^{2^d-1}_{k=0},\{\llbracket \mathbf{h}^{k,d}\rrbracket\}^{2^d-1}_{k=0},$
		\Statex \quad\quad\quad\quad\quad\: $F_d<t_d$). \label{4-line29}
		\EndFor
		
		\For{$k\in[0,2^D-1]$} \label{4-line21}
		\State $\llbracket{G}\rrbracket,\llbracket{H}\rrbracket \leftarrow \sum_{i=0}^{N-1}\llbracket\mathbf{g}^{k,D}_i\rrbracket,\sum_{i=0}^{N-1}\llbracket\mathbf{h}^{k,D}_i\rrbracket$.
		\State $\llbracket \mathbf{w}_{k}\rrbracket \leftarrow -\frac{\llbracket{G}\rrbracket}{\llbracket{H}\rrbracket+\llbracket\lambda\rrbracket}$.
		\EndFor \label{4-line25}
		\State Output a decision table $\mathcal{T}$ with $D$ tests, each held by a participant, and a secret-shared vector $\llbracket \mathbf{w}\rrbracket$ held by $P_1,P_2,\cdots,P_n$.
	\end{algorithmic}
\end{algorithm}

In this section, we introduce how {\main} combines the components introduced above to securely train a distributed decision table.
Algorithm \ref{alg::sectable} (named as $\mathsf{SecTable}$) describes this process. 
$\mathsf{SecTable}$ is the secure instantiation of Algorithm \ref{algo:plaintext-oblivious-tree} and relies on the coordination
of the secure components introduced above.

Algorithm \ref{alg::sectable} inputs the vertically partitioned datasets $\{\mathcal{D}^{N\times J_m}_m\}^n_{m=1}$, secret-shared label $\llbracket \mathbf{y} \rrbracket$, and aggregated inference results $\llbracket \mathbf{\hat{y}} \rrbracket$ from the previous round of training, and then outputs a distributed decision table. 
The Boolean tests at different levels of the decision table are held by different participants and the output values associated with each leaf node are stored in an secret-shared vector $\llbracket \mathbf{w}\rrbracket$.

At the beginning of Algorithm \ref{alg::sectable}, $P_1,P_2,\cdots,P_n$ calculate the secret-shared first and second-order gradient vectors (for the root node) (i.e., $\llbracket \mathbf{g} \rrbracket$ and $\llbracket \mathbf{h} \rrbracket$ at lines \ref{4-1}-\ref{4-6}). After calculating the secret-shared gradients, {\main} initializes a secret-shared vector $\llbracket \mathbf{w} \rrbracket=\llbracket \mathbf{0}_{2^D}\rrbracket$ to store the $2^D$ secret-shared output values.
After that, a decision table will be built level by level.
Similar to the functionality of $\mathsf{find\_split}$ in Algorithm \ref{algo:plaintext-oblivious-tree}, {\main} securely selects the optimal test $F_d<t_d$ at the $d$-th level ($d\in[0,D-1]$).
Specifically, the selection is made greedily: the learning algorithm first selects the best test for each feature (i.e., lines~\ref{4-line4}-\ref{4-line16}) and then selects the optimal test among the $J$ selected candidate tests (i.e., lines~\ref{4-line26}-\ref{4-line28}). 
It is noted that the selection here is made following that in the plaintext domain, which is introduced in Section \ref{gradient-boosting}, and the operations in the selection are substituted with secure operations and proposed components.
For the $J$ candidate tests, we initialize a secret-shared vector $\llbracket \boldsymbol{\sigma} \rrbracket=\llbracket \mathbf{0}_{J} \rrbracket$ to store the score of each feature's best test.
Besides, {\main} uses a public vector $\boldsymbol{\gamma}=\mathbf{0}_{J}$ to record the bucket ID of each feature's best test.

For simplicity, the total $J$ features are numbered from $0$ to $J-1$. 
At the beginning of the loop for the $j$-th feature in Algorithm \ref{alg::sectable}, the participant who holds the $j$-th feature first generates a permutation $\pi_j$ locally by sorting the values of the $j$-th feature (denoted by $\mathcal{D}(j)$) in the ascending order, which will be used to securely permute the secret-shared gradient vectors associated with each node at this level. 
After that, a naive method is to permute the gradient vectors and then adapt the Exact Greedy Algorithm \cite{xgboost} to enumerate each training sample to find the best test.
However, enumerating all training samples incurs heavy computation overhead.
Moreover, it will incur prohibitively expensive communication overhead in the distributed setting, degrading the efficiency of the system.

We propose a component $\mathsf{SecDisc}$ (shown in Algorithm \ref{alg::secdisc}) to tackle this challenge and enable {\main} to scale on larger datasets. 
Specifically, $\mathsf{SecDisc}$ inputs secret-shared gradient vectors $\llbracket{\mathbf g^{k,d}}\rrbracket$ and $\llbracket{\mathbf h^{k,d}}\rrbracket$ associated with the $k$-th node at the $d$-th level. 
The secret-shared gradients in $\llbracket{\mathbf g^{k,d}}\rrbracket$ and $\llbracket{\mathbf h^{k,d}}\rrbracket$ are securely discretized into $B$ buckets and stored in secret-shared vectors $\llbracket\boldsymbol{\alpha}^{k,d}\rrbracket$ and $\llbracket\boldsymbol{\beta}^{k,d}\rrbracket$, respectively. 
There are $B-1$ intervals among the $B$ buckets and each corresponds to a candidate test. In this way, {\main} only needs to select the best test from $B-1$ candidate tests for each feature, instead of enumerating $N$ samples, so as to save considerable computation and communication cost. In {\main}, we initialize a secret-shared vector $\llbracket \boldsymbol{\delta} \rrbracket=\llbracket \mathbf{0}_{B-1} \rrbracket$  for each feature to store the scores of the $B-1$ candidate tests.
%

After securely discretizing gradients into $B$ buckets, the $B-1$ candidate tests are evaluated to select the best test of the $j$-th feature. 
For the $c$-th candidate test ($c\in[0,B-2]$), the first $c+1$ buckets are aggregated to get $\llbracket{G_{l}}\rrbracket$ and $\llbracket{H_{l}}\rrbracket$, which are the sum of gradients associated with the left child node, and the remaining $B-c-1$ buckets are aggregated to get $\llbracket{G_{r}}\rrbracket$ and $\llbracket{H_{r}}\rrbracket$, which are the sum of gradients associated with the right child node.
The impurity of each node's two child nodes is securely computed following Eq. \ref{eq:imp} and then aggregated together to produce the secret-shared score of the  $c$-th candidate test following Eq. \ref{eq:score}. 
The division needed in Eq. \ref{eq:imp} can be securely calculated with our proposed secure component $\mathsf{SDiv}$ in Section \ref{secsigmoid}.
Then for the $j$-th feature, we will have $B-1$ secret-shared scores stored in $\llbracket \boldsymbol{\delta} \rrbracket$.

After getting the $B-1$ scores, we need to select the best test that achieves the minimum score, which requires a method to securely calculate the index of the minimum value in a secret-shared vector. 
To tackle this challenge, {\main} introduces a component $\mathsf{SecArgmin}$, which inputs a secret-shared vector and outputs the index of the minimum value of the vector.
It is noted that the key operation in the function Argmin is comparison, which is not naturally supported in the secret sharing domain. 
The secure comparison operation in our {\main} is introduced as follows.
Given two secret-shared values $\llbracket A\rrbracket$ and $\llbracket B\rrbracket$, {\main} first locally decomposes $\llbracket A- B\rrbracket$ into bits, and then inputs these bits into a parallel prefix adder (PPA) to securely compute the secret-shared most significant bit (MSB) of $\llbracket A- B\rrbracket$,  inspired by \cite{mohassel2018aby3,LiuZYY21}. 
After that, we convert the secret-shared MSB into the arithmetic sharing domain by the method in \cite{crypten}, so as to get the secret-shared result of the secure comparison. 
Based on the secure comparison method introduced above, $\mathsf{SecArgmin}$ inputs the secret-shared vector $\llbracket \boldsymbol{\delta} \rrbracket$ and then outputs bucket ID $q\in[0,B-2]$ of the $j$-th feature's best test in the plaintext. 
In {\main}, all participants can learn the produced bucket ID in the training stage, but only the participant who owns the $j$-th feature can get the threshold of the $j$-th feature's best test. 
Since the values of the $j$-th feature (denoted by $\mathcal{D}(j)$) is sorted in ascending order, the participant who owns the $j$-th feature can get the split threshold via looking up the sorted values (i.e., $\pi_j[\mathcal{D}(j)]$) with index $(q+1)\times(N/B)$.  
Other participants cannot deduce the split threshold because the $j$-th feature is kept locally by its owner  and unavailable to them.

To select the optimal test of all features, {\main}  lets the participants record the $j$-th feature's best bucket ID $q$ and secret-shared minimum split score $\llbracket \boldsymbol{\delta}_{q} \rrbracket$ at the $j$-th position of $\boldsymbol{\gamma}$ and $\llbracket \boldsymbol{\sigma}\rrbracket$, respectively (i.e., lines \ref{4-line23}-\ref{4-line24} in Algorithm \ref{alg::sectable}). 
Recall that for the $J$ features, we use $\boldsymbol{\gamma}$ and $\llbracket \boldsymbol{\sigma} \rrbracket$ to store the bucket ID and split score of each feature's best test. 
The $J$ indices of $\boldsymbol{\gamma}$ and $\llbracket \boldsymbol{\sigma} \rrbracket$ correspond to $J$ features, respectively. 
After enumerating $J$ features, $J$ scores are stored in $\llbracket \boldsymbol{\sigma} \rrbracket$ and $\mathsf{SecArgmin}$ is needed to be called again on $\llbracket \boldsymbol{\sigma} \rrbracket$. 
The output is the split feature $F_d$ of the optimal test, which is known by all participants. 
The optimal bucket ID $q_d$ of the split feature can then be retrieved with $F_d$ from $\boldsymbol{\gamma}$ (line \ref{4-line27}). After that, the participant who owns the split feature $F_d$ looks up the its sorted values $\pi_{F_d}[\mathcal{D}(F_d)]$ with index $(q_d+1)\times(N/B)$ to get the split threshold $t_d$.

After learning the $d$-th test $F_d<t_d$, the participant who owns $F_d<t_d$ cooperates with other participants to securely split all the nodes at the $d$-th level with $\mathsf{SecSplit}$ to create a new level (line \ref{4-line29}). A decision table in {\main} is learned level by level in this way. 
At the $D$-th level, {\main} securely calculates output values for the $2^{D}$ leaf nodes following Eq. \ref{eq:weight} (i.e., lines \ref{4-line21}-\ref{4-line25} in Algorithm \ref{alg::sectable}), where the division is securely calculated with $\mathsf{SDiv}$ in Section \ref{secsigmoid}. 
Finally, $\mathsf{SecTable}$ outputs a distributed decision table consisting of $D$ tests and $2^D$ secret-shared output values. 
Specifically, all participants know the split feature $F_d$ at the $d$-th level where $d\in[0,D-1]$, but each split threshold $t_d$ is only available to the participant who owns the feature $F_d$.

\subsection{Secure Distributed Decision Table Inference}
\label{SecInfer}

\begin{figure}[!t]
	\centering
	\includegraphics[scale=0.5]{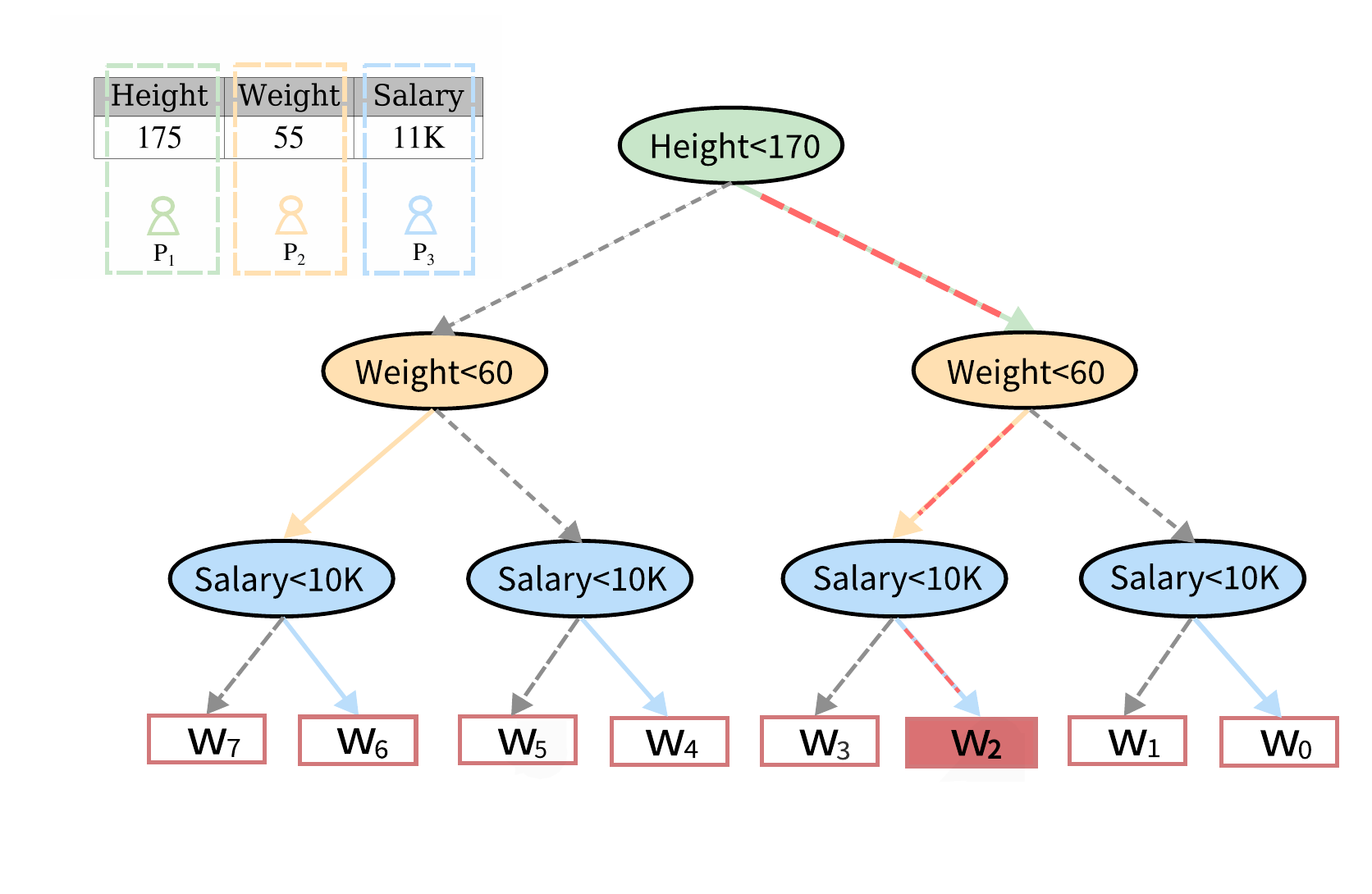}
	\caption{A simple example of secure distributed decision table inference. }
	\label{fig:SecInfer}
	\vspace{-10pt}
\end{figure}


In {\main}, each decision table in the ensemble learned in the secure training phase is held by the participants in a distributed manner, where each participant holds a part of it. 
Recall that in our secure VFL framework (Algorithm \ref{alg:secure-framework}), once a distributed decision table $\mathcal{T}$ is securely learned in a certain round, we need to perform secure inference over the training data using $\mathcal{T}$. 
The the produced inference results at this round will be securely aggreagted with previous inference results for use in securely training a new distributed decision table in the next round.
To prevent the partial model and local data on each participant from leaking during the secure inference process, we propose a secure distributed decision table inference protocol $\mathsf{SecInfer}$, which relies on secure multiplication of indicator vectors to conduct privacy-preserving inference, as shown in Algorithm \ref{alg::SecInfer}. 
$\mathsf{SecInfer}$ allows the participants to cooperatively perform secure inference on their local data utilizing the distributed ensemble and produce secret-shared  inference results while keeping the local data and partial model not unavailable to other participants throughout the inference process. We introduce the design of $\mathsf{SecInfer}$ as follows.

To securely produce the inference result of a vertically partitioned sample $\{\mathbf{x}^{J_1},\mathbf{x}^{J_2},\cdots,\mathbf{x}^{J_n}\}$, the participant $P_l$ who owns the test $F_d<t_d$ at the $d$-th level ($d\in[0,D-1]$) locally generates a leaf indicator (denoted by $\mathbf{u}_d$) by comparing the sample's feature value of $F_d$ (represented as $\mathbf{x}^{J_l}(F_d)$) with the split threshold $t_d$. After that, the leaf indicator $\mathbf{u}_d$ is secret-shared to other participants. 
The inference result could then be obliviously calculated by secure element-wise multiplication between the $D$ secret-shared leaf indicators $\{\llbracket\mathbf{u}_i\rrbracket\}_{i=0}^{D-1}$ and the secret-shared vector $\llbracket \mathbf{w}\rrbracket$ of decision table's output values.
We take a $3$-dimensional decision table to present the details of model distribution and $\mathsf{SecInfer}$ in Fig. \ref{fig:SecInfer}. 
Without loss of generality, we assume that the tests are ``$Height < 170$'', ``$Weight < 60$'', and ``$Salary < 10000$'',  held by participants $P_1,P_2,P_3$ , respectively, and the eight leaves' output values are secret-shared among $n$ participants. 
In Fig. \ref{fig:SecInfer}, it is noted that the exact test is only visible to the participant owning the corresponding feature, e.g., only $P_1$ knows the test ``$Height < 170$'' at the $0$-th level.

We take the inference of sample (``$Height = 175$'', ``$Weight = 55$'', ``$Salary = 11000$'') as an example. The three features are vertically partitioned and held by $P_1,P_2,P_3$, respectively.  For the first test ``$Height < 170$'', $P_1$ locally compares ``$Height = 175$'' with the threshold 170 and generate leaf indicator vector $\mathbf{u}_0=(0,0,0,0,1,1,1,1)$ to guide the inference path because $175>170$. Similarly, we can get $\mathbf{u}_1=(1,1,0,0,1,1,0,0)$ and $\mathbf{u}_2=(0,1,0,1,0,1,0,1)$. Each leaf indicator vector is then secret-shared among all participants. 
Recall that in the basic decision table inference introduced in Section \ref{sec:preliminaries}, the comparisons required by different Boolean tests are parallelized to accelerate inference due to the oblivious tree structure. Although the learned decision tables in the ensemble in {\main} are distributed and secret-shared, their oblivious structure remains unchanged, and thus operations at different levels can still be parallelized.
After the sharing of leaf indicator vectors, the inference result of this sample can be obliviously calculated by element-wise multiplication as follows: $\llbracket \mathbf{u}_0\rrbracket \times \llbracket \mathbf{u}_1\rrbracket\times \llbracket \mathbf{u}_2\rrbracket \times\llbracket \mathbf{w}\rrbracket$. In this way, each participant will not know which path in the distributed decision table is used during the secure inference process.
%
%

%

\begin{algorithm}[!t]
	\caption{Secure Decision Table Inference ($\mathsf{SecInfer}$)} 
	\label{alg::SecInfer}
	\begin{algorithmic}[1]
		
		\Require
		$P_1,P_2,\cdots,P_n$ hold a vertically partitioned sample $\{\mathbf{x}^{J_1},\mathbf{x}^{J_2},\cdots,\mathbf{x}^{J_n}\}$, a distributed decision table $\mathcal{T}$ of $D$ tests and a secret-shared vector $\llbracket \mathbf{w}\rrbracket$ containing $2^D$ output values;
		\Ensure
		$P_1,P_2,\cdots,P_n$ obtain the secret-shared inference result $\llbracket w\rrbracket$.
		
		\For{$d\in[0,D-1]$ \textbf{in parallel}}
		\State $P_l$ who holds the $d$-th test $F_d<t_d$ locally compares 
		\Statex \quad  $\mathbf{x}^{J_l}(F_d)$ with $t_d$. \label{6-line2}
		\State $P_l$ locally generates leaf indicator $\mathbf{u}_d$ at this level \label{6-line3}
		\Statex \quad  according to the outcome of the Boolean test. 
		\State $P_l$ secret-shares $\mathbf{u}_d$ to other participants. \label{6-line4}
		\EndFor 
		
		
		\leftline{//\underline{$P_1,P_2,\cdots,P_n$ perform:}}
		\State $\llbracket \mathbf{s}\rrbracket\leftarrow\llbracket \mathbf{u}_0\rrbracket \times \llbracket \mathbf{u}_1\rrbracket\times\cdots \times\llbracket \mathbf{u}_{D-1}\rrbracket\times\llbracket \mathbf{w}\rrbracket$.
		\State $\llbracket w\rrbracket \leftarrow \sum_{i=0}^{2^{D}-1} \llbracket \mathbf{s}_i\rrbracket$. \label{6-line7}
		\State Output the secret-shared inference result $\llbracket w\rrbracket$ held by $P_1,P_2,\cdots,P_n.$
	\end{algorithmic}
\end{algorithm}

\noindent \revise{\textbf{Remark.}} \revise{
	We note that there are some existing secure distributed decision tree inference methods \cite{tist,fang2021large} in the VFL setting. However, they are not well suited for the required secure distributed decision table inference in {\main}. 
	%
	%
	At a high level, these two works and {\main} share the common approach of using indicator vectors to enable secure inference.
	However, the inherent structural differences between decision tables and decision trees result in different methods of generating and utilizing indicator vectors for guiding the inference paths during privacy-preserving inference.  
	In secure distributed decision tree inference of \cite{tist,fang2021large}, each internal node is associated with an indicator vector, and the inference result is produced by secure multiplication of these indicator vectors. In contrast, for secure distributed decision table inference, each level in the decision table is associated with an indicator vector. 
	As a result, existing secure distributed decision tree inference methods cannot be efficiently extended to secure distributed decision table inference. 
	For instance, with our proposed $\mathsf{SecInfer}$ protocol and a decision table with a dimension of $D=6$, only six indicator vectors and six secure multiplications are needed. However, applying the method from \cite{tist,fang2021large} to our target problem would require $2^6-1$ indicator vectors and \emph{63} secure multiplications, resulting in poor efficiency.
	Moreover, the secure distributed tree inference method in \cite{fang2021large} is not applicable in our setting because it targets a two-party setting, while {\main} aims to support an arbitrary number of participants. 
}

\section{Security Analysis }
\label{sec:security-analysis}

\revise{{\main} utilizes standard secret sharing techniques \cite{demmler2015aby} to properly encrypt the intermediate information during both training and inference phases and the secret shares are uniformly distributed in a ring $\mathbb{Z}_{2^{Q}}$. 
In addition, throughout the VFL procedure, the feature data owned by each participant is kept locally. 
We follow the standard simulation-based paradigm \cite{lindell2017simulate} to analyze the security of {\main}.}
\revise{We first define the ideal functionality of our target privacy-preserving VFL with gradient boosted decision tables as follows.}
\begin{defn}	
	\emph{\revise{The ideal functionality $\mathcal{F}_{\text {{VDT}}}$ of privacy-preserving VFL with gradient boosted decision tables is formulated as follows:}} 
	
	\emph{\revise{-$\textbf{Input}$. The input to the $\mathcal{F}_{\text {{VDT}}}$  consists of datasets $\{\mathcal{D}^{N\times J_m}_m\}^n_{m=1}$ and the label set $\mathbf{y}$ from the participants $P_1,P_2,\cdots,P_n$.}}
	
	\emph{\revise{-$\textbf{Computation}$. Upon receiving the above input, the ideal functionality $\mathcal{F}_{\text {{VDT}}}$ performs training of gradient boosted decision tables and produces the trained model \textsf{VDT}, which consists of $T$ decision tables.}}
	
	\emph{\revise{-$\textbf{Output}$. The ideal functionality $\mathcal{F}_{\text {{VDT}}}$ broadcasts the split feature in the decision table to all participants, but only sends the split threshold to the participant who holds the corresponding split feature. 
			Additionally, the $\mathcal{F}_{\text {{VDT}}}$ splits the output values into $n$ secret shares and then distributes them to the participants $P_1,P_2,\cdots,P_n$.
	}}
\end{defn}

	

\begin{defn}
	\label{def:security}
	\emph{\revise{A protocol $\Pi$ securely realizes the ideal functionality $\mathcal{F}_{\text {{VDT}}}$ in the semi-honest adversary setting if a semi-honest participant does not learn any information about other participants' private data and partial model. Formally, let $\mathsf{View}^{\Pi}_{P_m}$ represent participant $P_m$'s view during the execution of $\Pi$. Formally, there should exist a PPT simulator, which can generate a simulated view $\mathsf{Sim}^{\Pi}_{P_m}$ such that $\mathsf{Sim}^{\Pi}_{P_m}$ is indistinguishable from $\mathsf{View}^{\Pi}_{P_m}$.}}
\end{defn}

\begin{thm}
	\revise{Our {\main} securely realizes
		the ideal functionality $\mathcal{F}_{\text {{VDT}}}$ against a semi-honest adversary who can statically corrupt a subset of $\tau$ participants ($\tau\leq n-1$) according to Definition \ref{def:security}.}
\end{thm}

\begin{proof}
	\revise{
		If the simulator for each sub-protocol exists, then our complete protocol is secure \cite{canetti2000security,katz2005handling,curran2019procsa}. 
		As presented before, {\main} consists of several secure sub-protocols: 1) secure division $\mathsf{SDiv}$; 2) secure Sigmoid $\mathsf{SecSigmoid}$; 3) secure node splitting $\mathsf{SecSplit}$; 4) secure distributed decision table inference $\mathsf{SecInfer}$;  5) secure discretization $\mathsf{SecDisc}$; 6) secure Argmin $\mathsf{SecArgmin}$.
		We use $\mathsf{Sim}^{\mathsf{X}}_{P_m}$ as the
		simulator which can generate $P_m$'s view in sub-protocol $\mathsf{X}$ $ (\mathsf{X}\in\{\mathsf{SDiv}, \mathsf{SecSigmoid}, \mathsf{SecSplit}, \mathsf{SecInfer}, \mathsf{SecDisc}, \mathsf{SecArgmin}\})$ on corresponding input and output.
		Obviously, the simulators for $\mathsf{X}\in\{\mathsf{SDiv, SecSigmoid,SecArgmin}\}$ must exist, because they are comprised of basic operations in the secret sharing domain\cite{demmler2015aby}. 
		In the execution of these three protocols, even if a subset of $\tau$ participants is corrupted, the honest participants' private data will not be leaked due to the security guarantee of additive secret sharing\cite{crypten}.
		In addition, when revealing the index of the minimum value in a secret-shared vector to all participants in $\mathsf{SecArgmin}$, the simulator can adjust the shares of the result such that the revealed index is indeed the value received from the ideal functionality, and thus the simulator of $\mathsf{SecArgmin}$ exists. 
		Therefore, {\main} is secure if the simulators for the remaining sub-protocols exist, i.e., $\mathsf{SecSplit}$ in Section \ref{secsplit}, $\mathsf{SecInfer}$ in Section \ref{SecInfer}, $\mathsf{SecDisc}$ in Section \ref{secdisc}. 
		We next provide the existence of the simulators for the remaining sub-protocols.
	}
\end{proof}

\begin{thm}
	\revise{The simulators for sub-protocols $\mathsf{SecSplit}$ and $\mathsf{SecInfer}$ exsit.}
\end{thm}

\begin{proof}
	\revise{The sub-protocols $\mathsf{SecSplit}$ and $\mathsf{SecInfer}$ both require that the participant who holds the split generates indicator vectors locally and then secret-shares them to other participants. 
		For simplicity, we assume that $P_l$ holds the split threshold and collaborates with $P_{m\neq l}$ to split the node and conduct inference.
		We then prove the existence of the simulators $\mathsf{Sim}^{\mathsf{SecSplit}}_{P_l}$, $\mathsf{Sim}^{\mathsf{SecInfer}}_{P_l}$ for $P_l$ and the simulators $\mathsf{Sim}^{\mathsf{SecSplit}}_{P_{m\neq l}}$, $\mathsf{Sim}^{\mathsf{SecInfer}}_{P_{m\neq l}}$ for $P_{m\neq l}$ due to their different computation.
		We also note that there is no difference in the simulation between AP and PP due to their role equivalence in the execution of $\mathsf{SecSplit}$ and $\mathsf{SecInfer}$}.

	\begin{itemize}
		
		\item \revise{ $\mathsf{Sim}^{\mathsf{SecSplit}}_{P_l}$, $\mathsf{Sim}^{\mathsf{SecInfer}}_{P_l}$ for $P_{l}$: The simulator is simple since $P_l$ only secret-shares its local indicator vectors and perform secure multiplications on the secret-shared vectors.
			In the execution of $\mathsf{SecSplit}$ and $\mathsf{SecInfer}$, $P_l$ receives nothing. 
			Moreover, since secure multiplication is the basic operation in the secret sharing domain, the privacy of the honest participants' data is ensured even if $P_l$ colludes with $\tau-1$ participants.
			Therefore, it is clear that the simulated view is indistinguishable from the real view.
		}	
		
		\item \revise{ $\mathsf{Sim}^{\mathsf{SecSplit}}_{P_{m\neq l}} $, $\mathsf{Sim}^{\mathsf{SecInfer}}_{P_{m\neq l}}$ for $P_{m\neq l}$: In the execution of $\mathsf{SecSplit}$ and $\mathsf{SecInfer}$, the only information $P_{m\neq l}$ receives is the secret share (denoted by $\langle\mathbf{r}\rangle_{m\neq l}$) of $P_l$'s indicator vector $\mathbf{r}$. The secret share $\langle\mathbf{r}\rangle_{m\neq l}$ is randomly generated at $P_l$ and thus is uniformly random in $P_{m\neq l}$'s view. Therefore, the distribution over the real secret share $\langle\mathbf{r}\rangle_{m\neq l}$ received by $P_{m\neq l}$ in the execution and over the simulated $\langle\mathbf{r}\rangle_{m\neq l}$ generated by the simulator are identically distributed. 
			Furthermore, the secure multiplication involved is the basic operation in the secret sharing domain, which means $P_{m\neq l}$ learns no additional information from the execution, even if $P_{m\neq l}$ colludes with $\tau-1$ participants.
			Therefore, the simulated view is indistinguishable from the real view.
		}
	\end{itemize}
	
\end{proof}

\begin{thm}
	\revise{
		The simulator for the sub-protocol $\mathsf{SecDisc}$ exists.  
	}
\end{thm}
\begin{proof}
	\revise{
		The security of $\mathsf{SecDisc}$ relies on the security of its sub-protocol $\mathsf{SecPerm}$ because the operations in $\mathsf{SecDisc}$ besides $\mathsf{SecPerm}$ are basic operations in the secret sharing domain. 
		Therefore, if the protocol $\mathsf{SecPerm}$ can be simulated, the existence of a simulator for the protocol $\mathsf{SecDisc}$ follows. 
		In $\mathsf{SecPerm}$, for simplicity, we assume that $P_l$ sorts the values for a feature locally to generate a  permutation $\pi$. 
		We then prove the existence of the simulator $\mathsf{Sim}^{\mathsf{SecDisc}}_{P_l}$ for $P_l$ and the simulator $\mathsf{Sim}^{\mathsf{SecDisc}}_{P_{m\neq l}}$ for $P_{m\neq l}$ due to their different computation. Similarly, since the role equivalence of AP and PP in the execution of $\mathsf{SecDisc}$, there is no difference in the analysis for them.
	}
	\begin{itemize}
		\item \revise{ $\mathsf{Sim}^{\mathsf{SecDisc}}_{P_l}$ for $P_l$: 
			In the execution of $\mathsf{SecDisc}$, ${P_l}$ only receives secret share $\langle\mathbf{x-r}\rangle_{m\neq l}$ of $\llbracket\mathbf{x-r}\rrbracket$ from ${P_{m\neq l}}$ to reconstruct $\mathbf{x-r}$. In the simulated view, $P_l$ receives $n-1$ random vectors. Therefore, we need to prove that $\langle\mathbf{x-r}\rangle_{m\neq l}$ is uniformly random in the view of $P_l$. Obviously, the above claim is valid, because $\langle\mathbf{x-r}\rangle_{m\neq l}$ is a random vector generated at ${P_{m\neq l}}$, which must be uniformly random in the view of $P_l$. Therefore the distributions over the real $\langle\mathbf{x-r}\rangle_{m\neq l}$ received by $P_l$ in the protocol execution and over the simulated $\langle\mathbf{x-r}\rangle_{m\neq l}$ are identically distributed. Moreover, even if $P_l$ colludes with $\tau-1$ participants, $P_l$ still only learns a randomly masked version of $\mathbf{x}$, thus $\mathbf{x}$ would not be leaked to $P_l$. 
			Thus, the simulator for $\mathsf{Sim}^{\mathsf{SecPerm}}_{P_{l}}$ exists, which indicates that the simulator for $\mathsf{Sim}^{\mathsf{SecDisc}}_{P_{l}}$ also exists.
		}

		\item \revise{ $\mathsf{Sim}^{\mathsf{SecDisc}}_{P_{m\neq l}}$ for $P_{m\neq l}$: In the execution of $\mathsf{SecDisc}$, ${P_{m\neq l}}$ only receives permutation $\pi_s$ from $P_l$. The permutation $\pi_s$ is randomly generated at $P_l$ and is uniformly random in $P_{m\neq l}$'s view. Therefore, the distribution over the real permutation $\pi_s$ received by $P_{m\neq l}$ in the execution and over the simulated $\pi_s$ generated by the simulator is identically distributed. 
			In case of corruption, there are two situations: 1) If participant $P_l$ is corrupted, nothing is revealed regarding the vector $\mathbf{x}$.
			2) If participant $P_l$ is not corrupted, the private permutation $\pi$ on $P_l$ and the private vector $\mathbf{x}$ are protected, even if all other participants collude. This security guarantee comes from the fact that the corrupted participants could learn nothing about $\pi_p$. 
                Without the knowledge of $\pi_p$, $\pi$ cannot be deduced because $\pi(\cdot)=\pi_s[\pi_p(\cdot)]$.
			Thus, the simulator for $\mathsf{Sim}^{\mathsf{SecPerm}}_{P_{m\neq l}}$ exists, which indicates that the simulator for $\mathsf{Sim}^{\mathsf{SecDisc}}_{P_{m\neq l}}$ also exists.}
	\end{itemize}

\end{proof}

\section{Experiments}

\label{sec:experiments}
\begin{table}[!t]
	\centering
	\caption{\revise{Statistics of Datasets}}
	\label{tab:statistics}
	\setlength{\tabcolsep}{1.5mm}{
		\begin{tabular}{lcccc}
			\toprule 
			Dataset Type & Name & Samples & Features & Task \\
			\midrule
			Real-world &
			Cal Housing\tablefootnote{\url{https://scikit-learn.org/stable/modules/generated/sklearn.datasets.fetch_california_housing.html}} & 20,640 & 8 &Regression \\
			&Credit\tablefootnote{\url{https://www.kaggle.com/datasets/uciml/default-of-credit-card-clients-dataset}} & 30,000 & 23 &Classification \\
			&Breast Cancer\tablefootnote{\url{https://www.kaggle.com/datasets/uciml/breast-cancer-wisconsin-data}} & 570 & 30 &Classification \\
			\midrule
			Synthetic &
			Syn$^A$ & 20,000 & 80 &Regression \\
			&Syn$^B$ & 20,000 & 60 &Regression \\
			&Syn$^C$ & 20,000 & 40 &Classification \\
			\bottomrule
			
	\end{tabular}}
	
\end{table}

\subsection{Setup}
\revise{
We implement our protocols in Python. All experiments are performed on a workstation with 16 Intel I7- 10700K cores, 64GB RAM, and 1TB SSD external storage, running Ubuntu 20.04.2 LTS. It is worth mentioning that the practice of evaluating VFL algorithms on a single machine also exists in prior works \cite{cheng2021secureboost,tist,fang2021large,Opboost}. 
We also note that in practice, it is not common to have VFL scenarios with more than four participants since it could be hard to bring together many enterprises \cite{fu2021vf2boost,jin2021cafe}. Therefore, in our experiments, we follow prior works\cite{tist,jin2021cafe,fu2021vf2boost,feverless} to conduct experiments with four participants. 
The communication between  participants on the workstation is emulated by the loopback filesystem, where the delay is set to 5 ms and bandwidth is set to 100 Mbps. 
}

\noindent$\textbf{Datasets.}$ \revise{We use three real-world datasets to evaluate the accuracy and efficiency of {\main} and three synthetic datasets to further evaluate the scalability of {\main}. The synthetic datasets are generated with \textit{sklearn}\footnote{\url{https://scikit-learn.org/stable/}} library. Table \ref{tab:statistics} summarizes the statistics of the six datasets.
}
We divide each dataset into two parts for training and testing respectively according to the ratio of 8:2.	In addition to the preprocessing in the above, each dataset is split and distributed to all participants vertically and evenly. Similar to previous works on VFL \cite{tian2020federboost, fang2021large}, we assume that the records in each participant's database have been properly aligned beforehand. 

\noindent$\textbf{Protocol instantiation.}$ Our protocols are instantiated using the following parameter settings. We use the ring $\mathbb{Z}_{2^{64}}$ with the number of precision bits $l=20$. The number of iterations for the reciprocal approximation is set to 20 (with $1/Y=1/2^{20}$ as the initialization). For approximating the exponential function, we set $2^n=4$, thus the approximation requires $\log 2^n = 2$ rounds of multiplication.
The hyper-parameters are public to all participants. We fix $\lambda=1$ of Eq. \ref{eq:imp} and vary the public parameter $T$ (the number of decision tables), $B$ (the number of buckets), and $D$ (the dimension of decision tables) in our experiments to demonstrate the utility, efficiency, and scalability of {\main}.

\revise{
	It is noted that to handle real numbers for secure computation, we follow the common practice of fixed-point representation, where real numbers are scaled by a factor of $2^l$ ($l$ represents the number of precision bits) and then rounded. 
	As a result, when two scaled values are multiplied, the result is under a scaling factor of $2^{2l}$. 
	Therefore, truncation is required to scale down the multiplication result, making its scaling factor $2^l$ again.
	{\main} resorts to the method in \cite{crypten} to support secure truncation on a result $\llbracket z \rrbracket$ produced from secret-shared multiplication, which works as follows.
	Firstly, the secret sharing of the number of wraps (denoted by $\llbracket \theta_z \rrbracket$) in $z$ needs to be computed, which is subject to $ \theta_z= (\sum^n_{m=1}\langle z \rangle_m-z)/2^Q$.  
	It is noted that to count the number of wrap rounds in this truncation protocol, computing the sum of shares (represented in the form of $\sum^n_{m=1}\langle z \rangle_m$) does not involve modular arithmetic. 
	The computation of $\llbracket \theta_z \rrbracket$ proceeds as follows. 
	All parties hold in advance a secret-shared random value $\llbracket r \rrbracket$ and its secret-shared wrap count $\llbracket \theta_r \rrbracket$ subject to $\theta_r=(\sum^n_{m=1}\langle r \rangle_m- r)/2^Q$.  
	For secure truncation, all parties first compute $\llbracket p \rrbracket=\llbracket z+r \rrbracket$. 
	After that, each party $P_m$ computes the differential wraps produced between its shares of $\llbracket z \rrbracket$, $\llbracket r \rrbracket$, and $\llbracket p \rrbracket$: $\langle\beta_{zr}\rangle_m=(\langle z \rangle_m+\langle r \rangle_m-\langle p \rangle_m)/2^Q$, where no modulo operation is required in calculating $\langle z \rangle_m+\langle r \rangle_m-\langle p \rangle_m$.
	Next, $P_{m\neq1}$ sends $\langle p \rangle_{m\neq1}$ to $P_1$ to reconstruct $p$  and $P_1$ also computes $\theta_p=(\sum_{m=1}^{n}\langle p \rangle_m-p)/2^Q$. 
	Finally, each $P_m$ produces the secret share $\langle\theta_z\rangle _m = j\times\theta_p + \langle\beta_{zr}\rangle_m-\langle\theta_{r}\rangle_m$, where $j=1$ if $m=1$ and $j=0$ if $m\neq1$. 
	Then $\llbracket \theta_z \rrbracket$ is used to correct the truncation: $\llbracket z \rrbracket = \frac{\llbracket z \rrbracket - \llbracket \theta_z \rrbracket 2^Q}{2^l}$.
	The above method only needs 1 online communication round to compute the number of wraps $\llbracket \theta_z \rrbracket$, making it efficient and practical for use in {\main}. 
}

\begin{table}[!t]
	\centering
	\caption{Accuracy Comparison}
	\label{tab:acc}
	\begin{tabular}{lccc}
		\toprule
		{Dataset} & {Method} & {RMSE} & {ACC/AUC} \\
		\midrule
		Cal Housing & Plaintext & 0.51 & - \\
		& {\main} & 0.51 & - \\
		\hline
		Credit & Plaintext & - & 80.3\%/0.7438 \\
		& {\main} & - & 80.3\%/0.7433 \\
		\hline
		Breast Cancer & Plaintext & - & 96.5\%/0.999 \\
		& {\main} & - &  96.5\%/0.998 \\
		\bottomrule
	\end{tabular}
\end{table}

\begin{table}[!t]
	
	\centering
	\caption{{\main}'s Computation and Communication Performance}
	
	\setlength{\tabcolsep}{5mm}{
		\begin{tabular}{cccc}
			
			\toprule Dataset & \multicolumn{2}{c}{ Online Secure Training } \\
			\hline 
			& { Time (seconds) } & { Comm. (GB) } \\
			\hline
			Cal Housing & 6333  & $41.1$  \\
			Credit & 1724  & $11.8$  \\
			Breast Cancer & 486 & $0.54$  \\
			\bottomrule
			
	\end{tabular}}
	\label{tab:efficiency}
\end{table}

\subsection{Utility Evaluation}
\label{subsec:utility_evaluation}

We first compare the accuracy of two approaches: our {\main} and plaintext centralized learning of gradient boosted decision tables. 
%
%
For the Cal Housing dataset, we set $T=50$ to build $50$ decision tables in the ensemble model and $5$ to limit the dimension of each decision table. The number of buckets $B$ is set as $32$. For the Credit dataset, we set $T=10$, $D=4$ and $B=32$. For the Breast Cancer dataset, we set $T=10$, $D=3$ and $B=32$.
Following other works on gradient boosting \cite{fang2021large}, more decision tables are trained for regression tasks to guarantee accuracy.
We use the identical parameters in {\main} and plaintext.
For the regression tasks, we use the Root Mean Square Error (RMSE) as the evaluation metric. 
For the evaluation of classification tasks, we report the results using two commonly used metrics: Accuracy (ACC) and Area Under the ROC Curve (AUC). 
The accuracy of {\main} and plaintext on both regression and classification tasks are reported in Table \ref{tab:acc}.

Fig. \ref{fig:curve} shows the RMSE/test error on the three real-world datasets, for varying number of decision tables. 
Note that the test error is defined as the complement of the Accuracy (i.e., $1-$ACC). 
It is observed that the difference in utility between {\main} and plaintext is obvious at the very beginning, but the difference rapidly diminishes as the increase of number of decision tables in the ensemble. 
This indicates the similar convergence behavior of {\main} and plaintext.
From the above results, we can conclude that our {\main} achieves compatible accuracy with plaintext centralized learning of gradient boosted decision tables on both classification and regression tasks.

\begin{figure}[!t]
	\centering
	\begin{minipage}[t]{0.3333\linewidth}
		\centering
		\includegraphics[width=\linewidth]{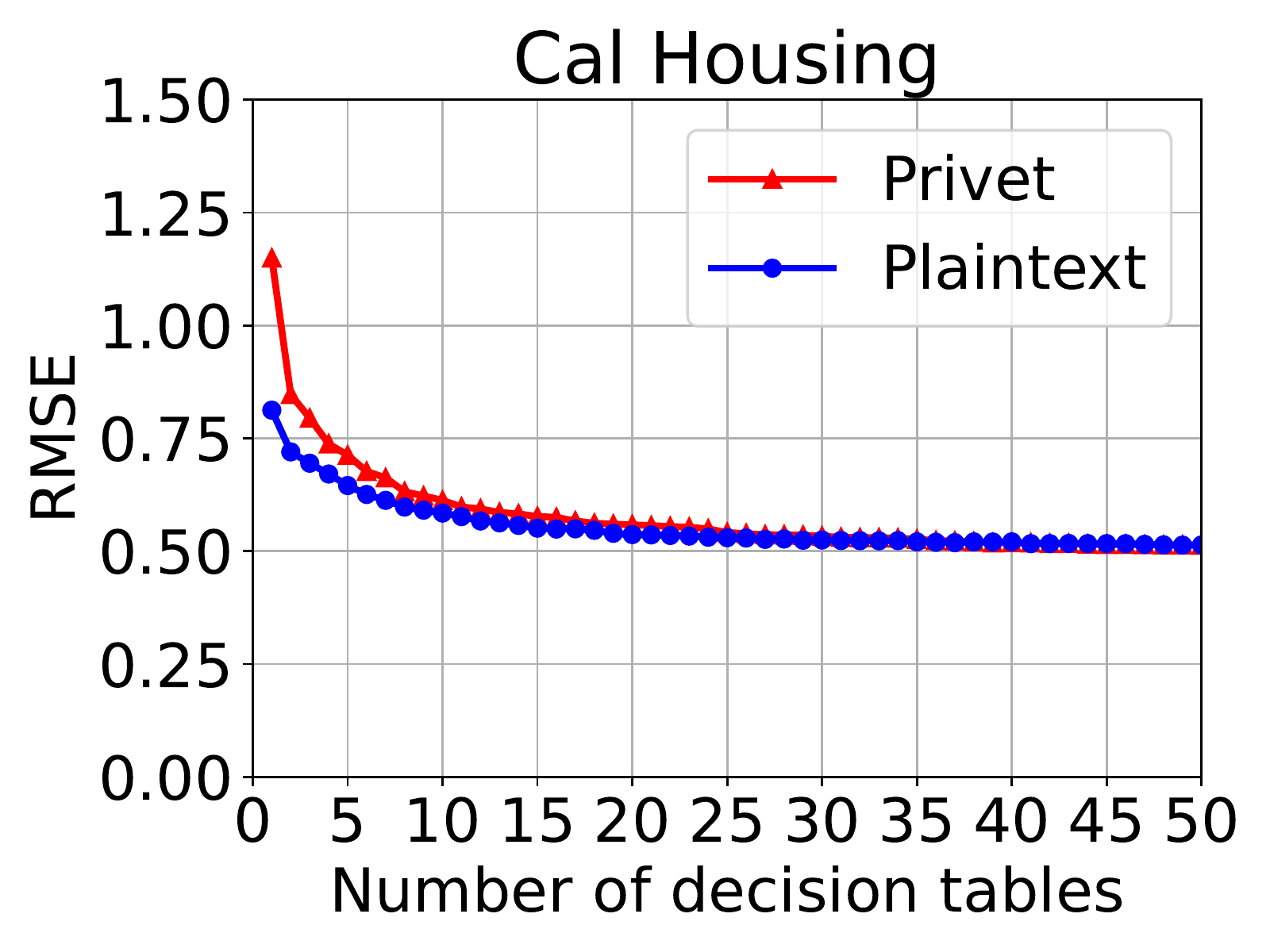}
	\end{minipage}%
	\begin{minipage}[t]{0.3333\linewidth}
		\centering
		\includegraphics[width=\linewidth]{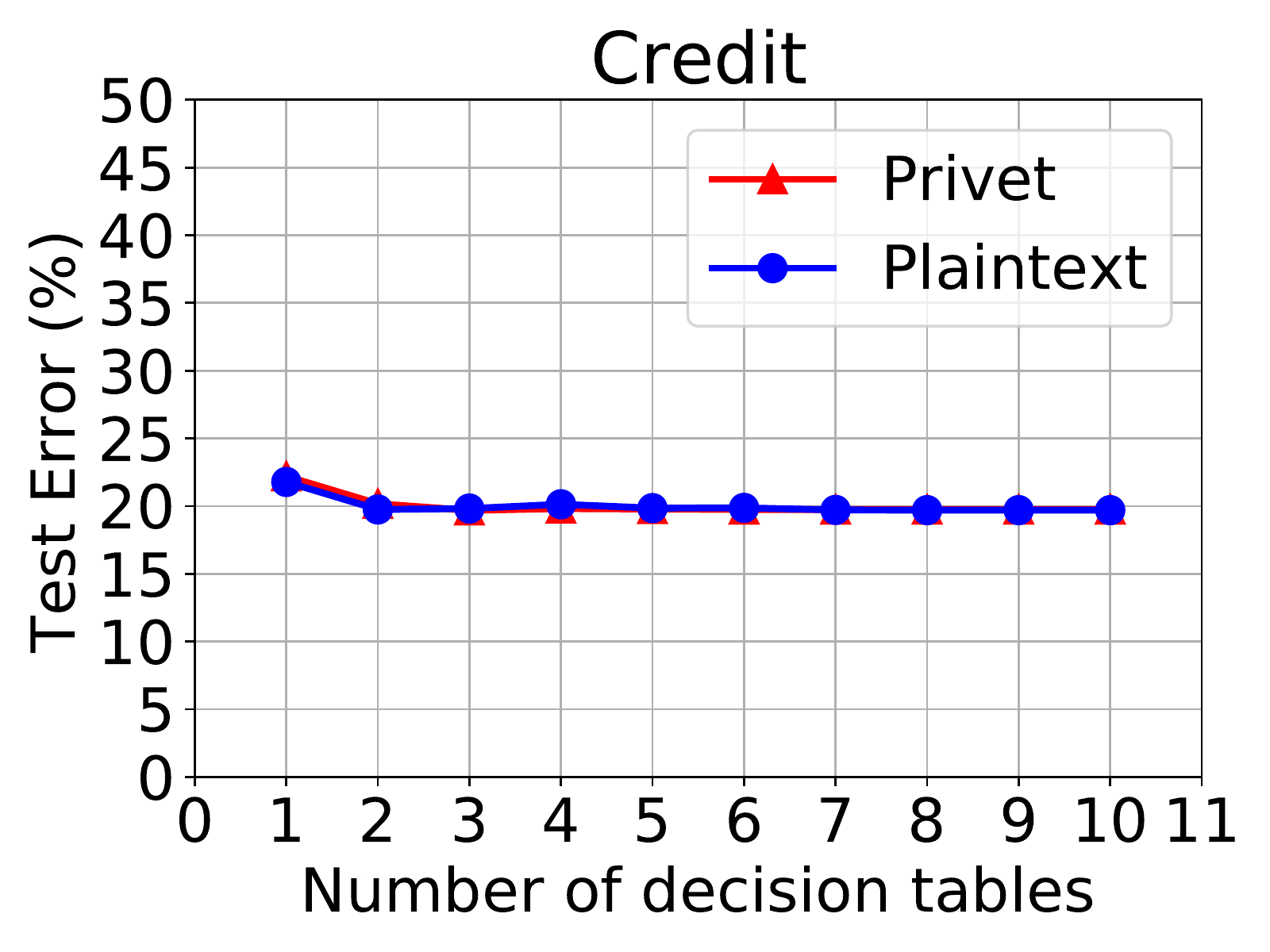}
	\end{minipage}%
	\begin{minipage}[t]{0.3333\linewidth}
		\centering
		\includegraphics[width=\linewidth]{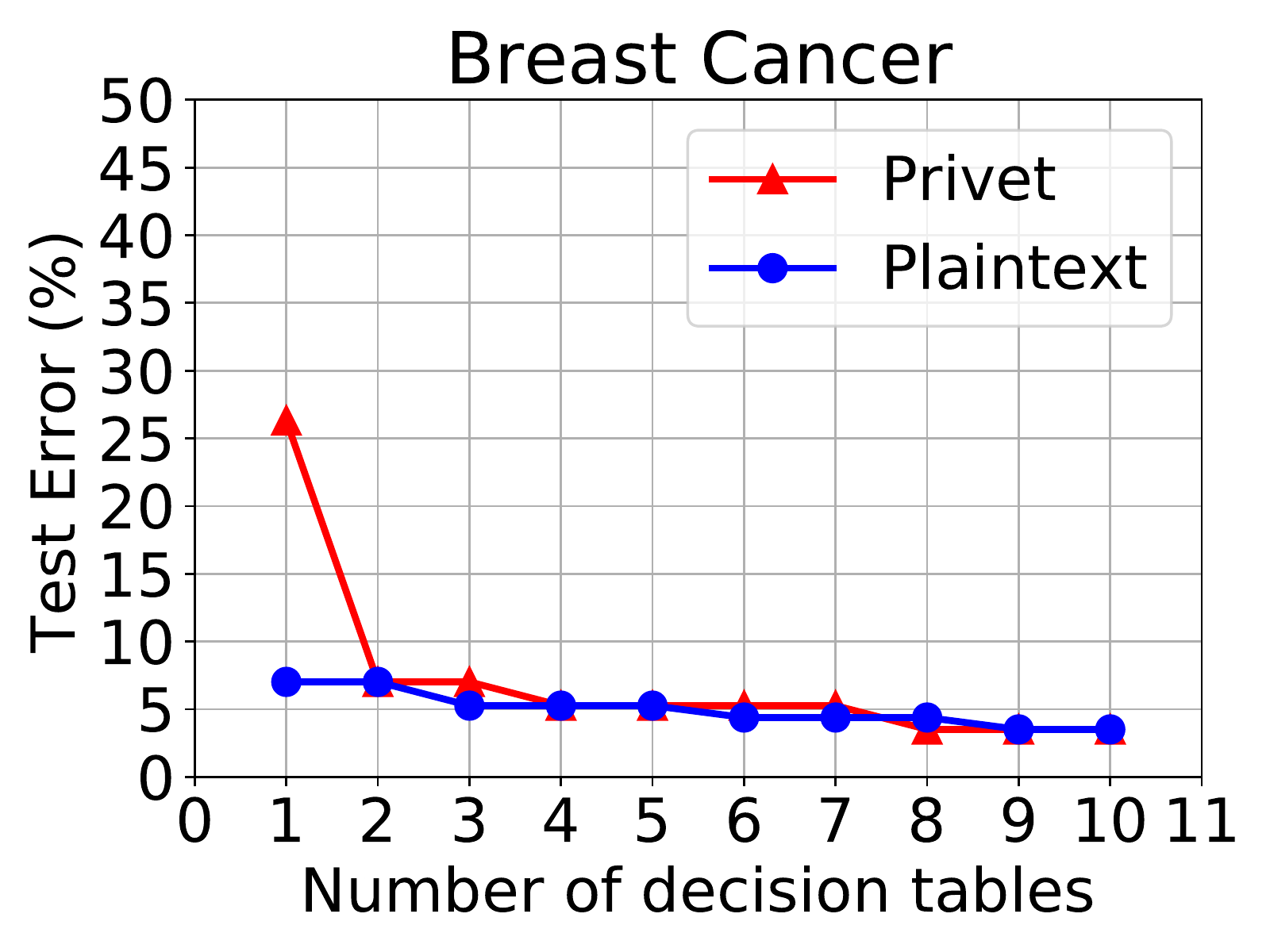}
	\end{minipage}
	\centering
	\vspace{-5pt}
	\caption{RMSE/test error on the three real-world datasets, for different numbers of decision tables.}
	\label{fig:curve}
\end{figure}

\begin{figure}[t!]
\centering
	\begin{minipage}[t]{0.42\linewidth}
		\centering
		\includegraphics[width=\linewidth]{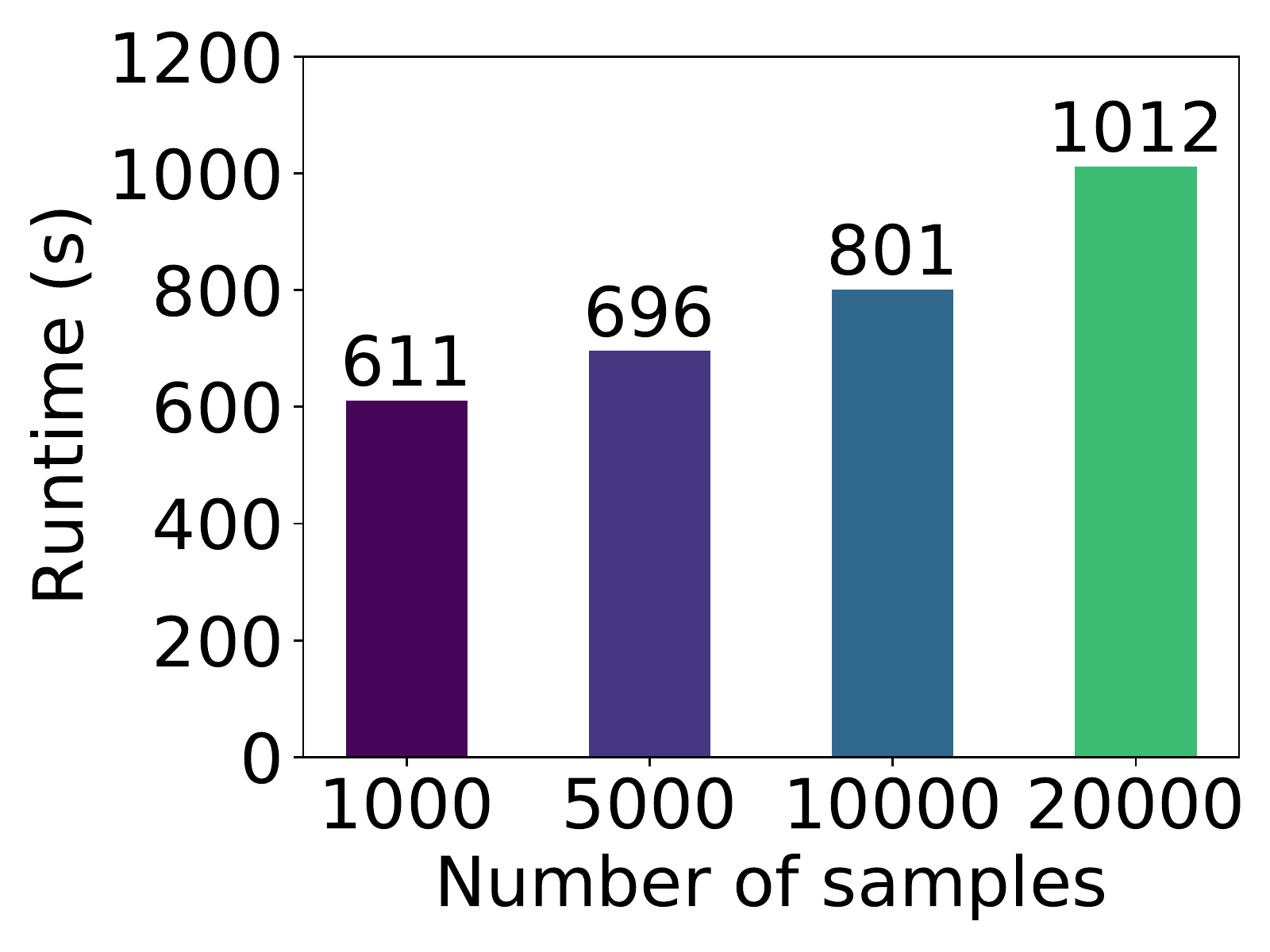}\\(a)
	\end{minipage}
	\begin{minipage}[t]{0.45\linewidth}
		\centering
		\includegraphics[width=\linewidth,height=1.1in]{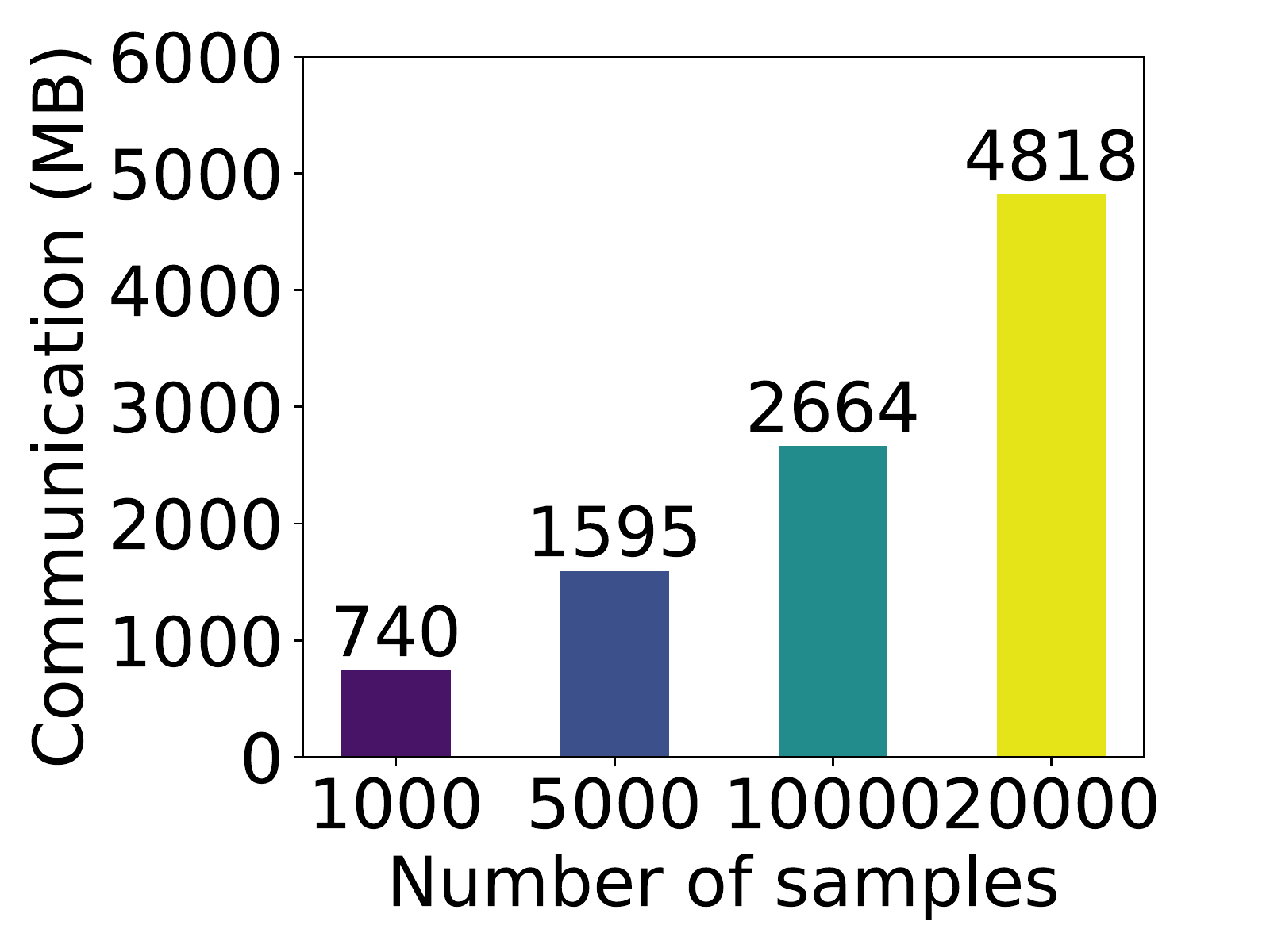}\\(b)
	\end{minipage}
	\caption{\revise{Performance on Syn$^A$ dataset with 80 features for different numbers of training samples (with the dimension of decision tables $D=2$, the number of buckets $B=32$ and the number of decision tables $T=10$): (a) Runtime; (b) Communication.}}
	\label{fig:runtime-commn-varying-sample}
	\vspace{-5pt}
\end{figure}

\begin{figure}[t!]
\centering
	\begin{minipage}[t]{0.42\linewidth}
		\centering
		\includegraphics[width=\linewidth]{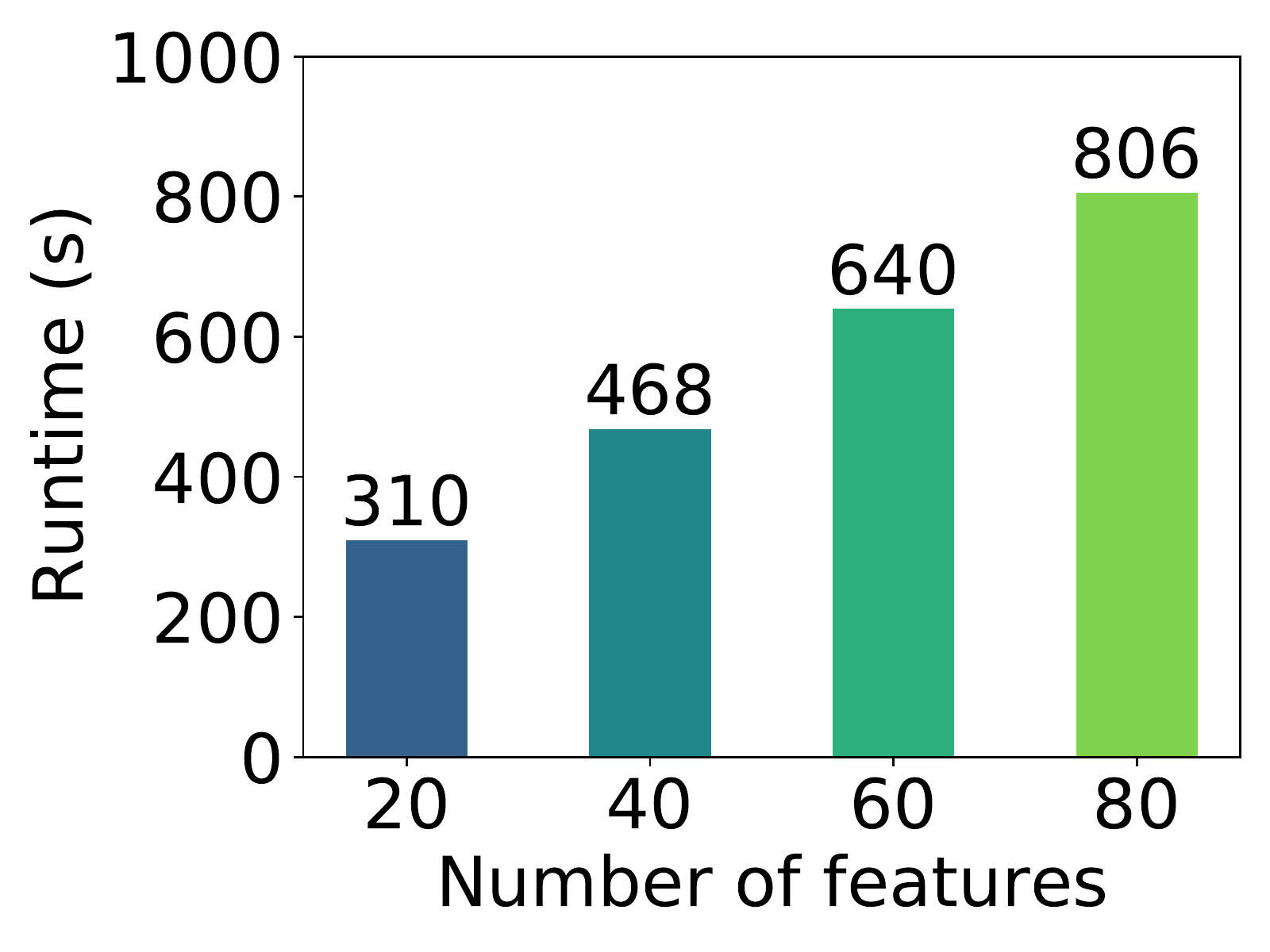}\\(a)
	\end{minipage}
	\begin{minipage}[t]{0.45\linewidth}
		\centering
		\includegraphics[width=\linewidth,height=1.1in]{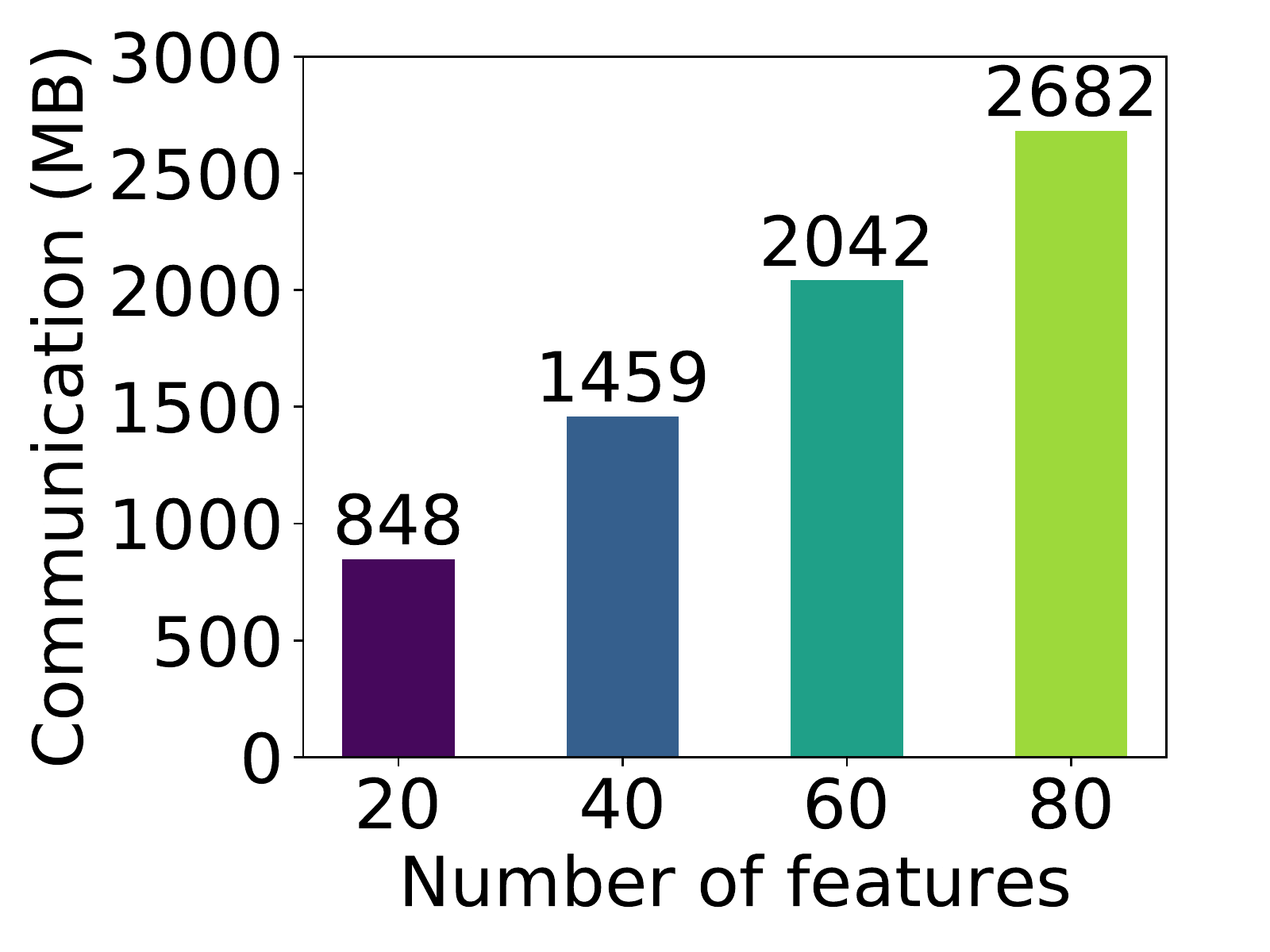}\\(b)
	\end{minipage}
	\caption{\revise{Performance on Syn$^A$ dataset with 10000 samples for different numbers of features (with the dimension of decision tables $D=2$, the number of buckets $B=32$ and the number of decision tables $T=10$): (a) Runtime; (b) Communication.}}
	\label{fig:runtime-commn-varying-feature}
	\vspace{-5pt}
\end{figure}

\begin{figure}[t!]
\centering
	\begin{minipage}[t]{0.32\linewidth}
		\centering
		\includegraphics[width=\linewidth]{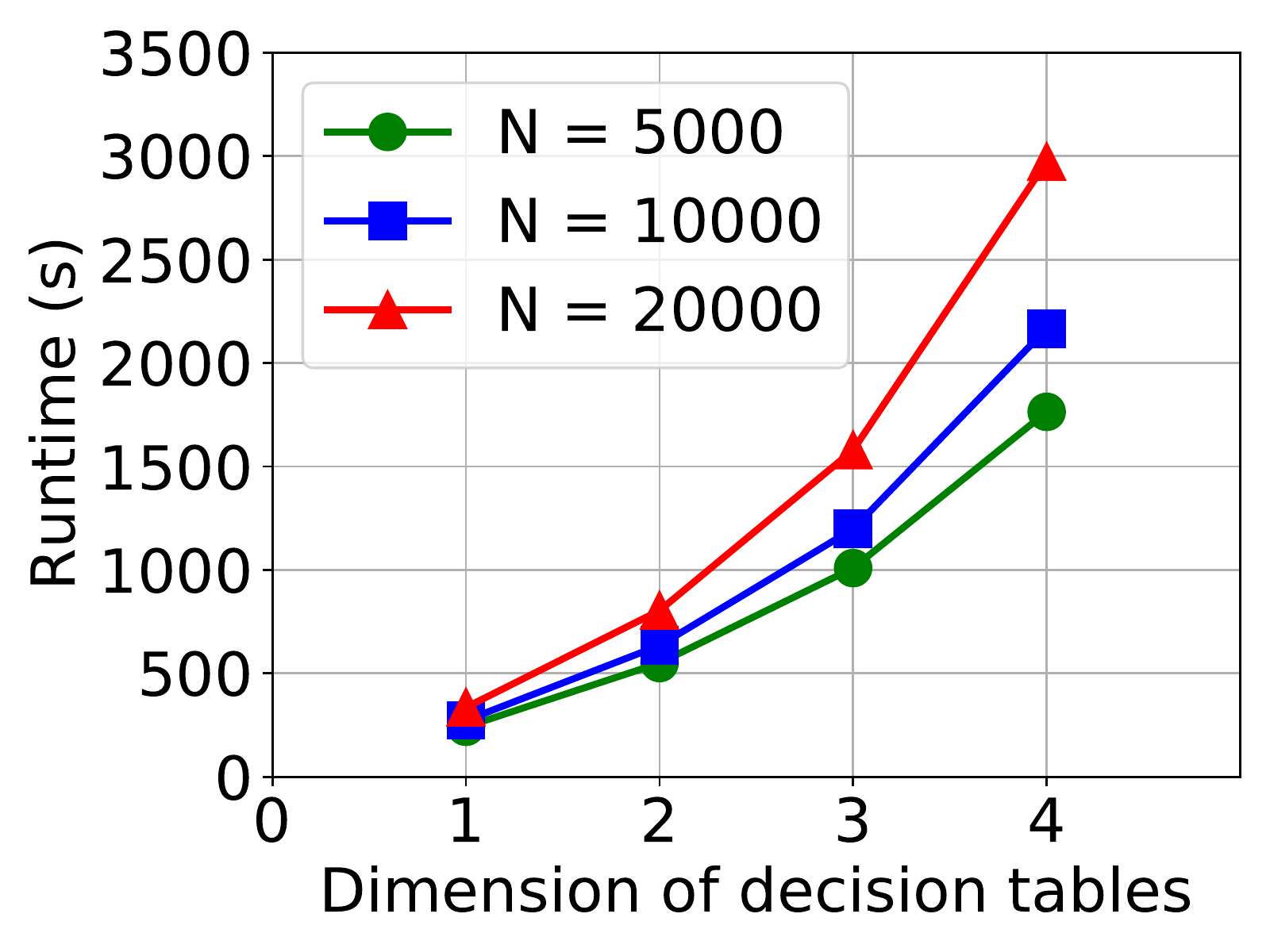}\\(a)
	\end{minipage}
	\begin{minipage}[t]{0.32\linewidth}
		\centering
		\includegraphics[width=\linewidth]{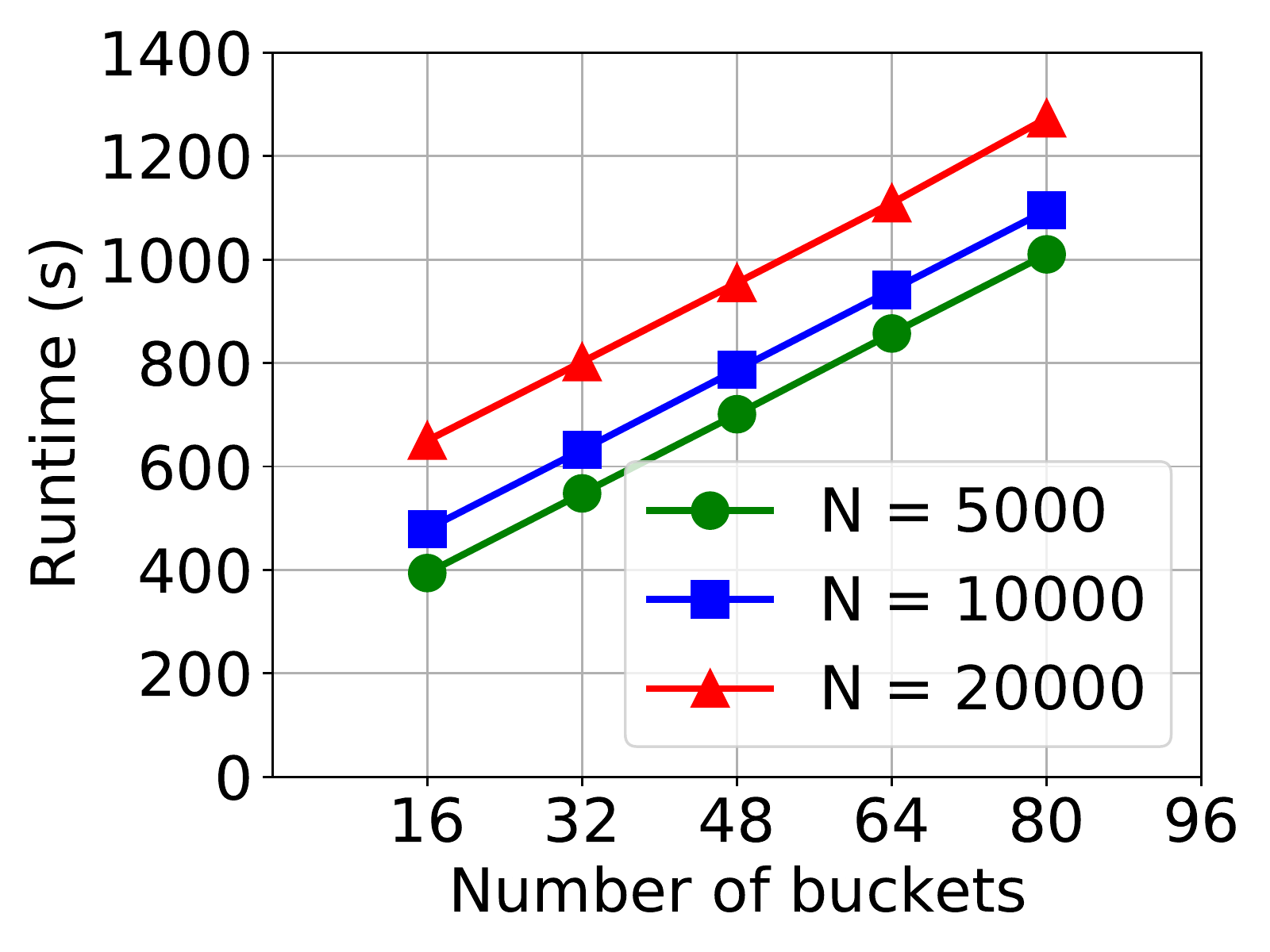}\\(b)
	\end{minipage}
	\begin{minipage}[t]{0.32\linewidth}
		\centering
		\includegraphics[width=\linewidth]{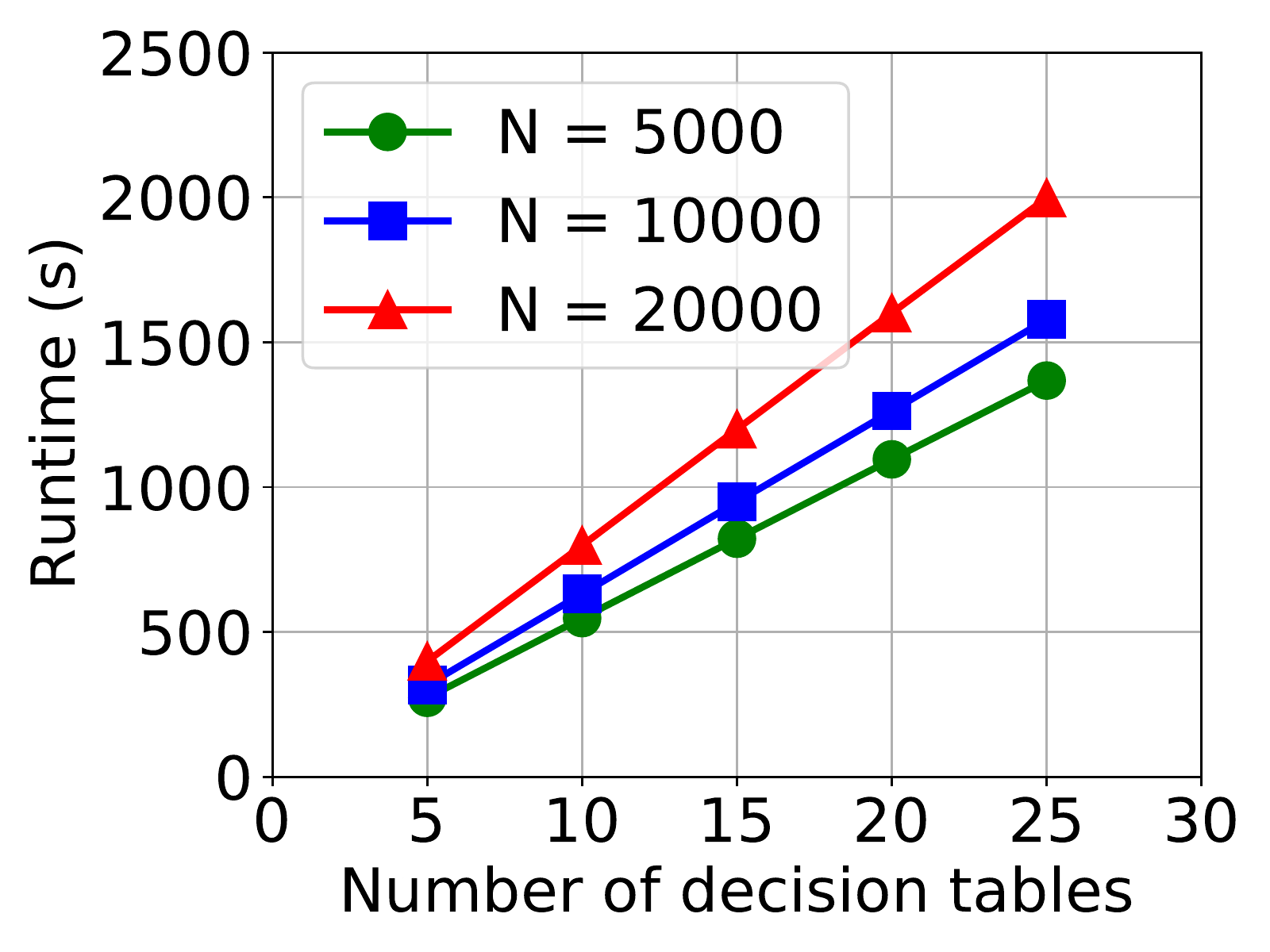}\\(c)
	\end{minipage}
	\caption{\revise{Runtime performance on Syn$^B$ with 60 features for different numbers of training samples $N$ and (a) varying dimension of decision tables $D$, (b) varying number of buckets $B$, and (c) varying number of decision tables $T$, respectively.}}
	\label{fig:reg_scalability}
	\vspace{-5pt}
\end{figure}

\begin{figure}[t!]
\centering
	\begin{minipage}[t]{0.32\linewidth}
		\centering
		\includegraphics[width=\linewidth]{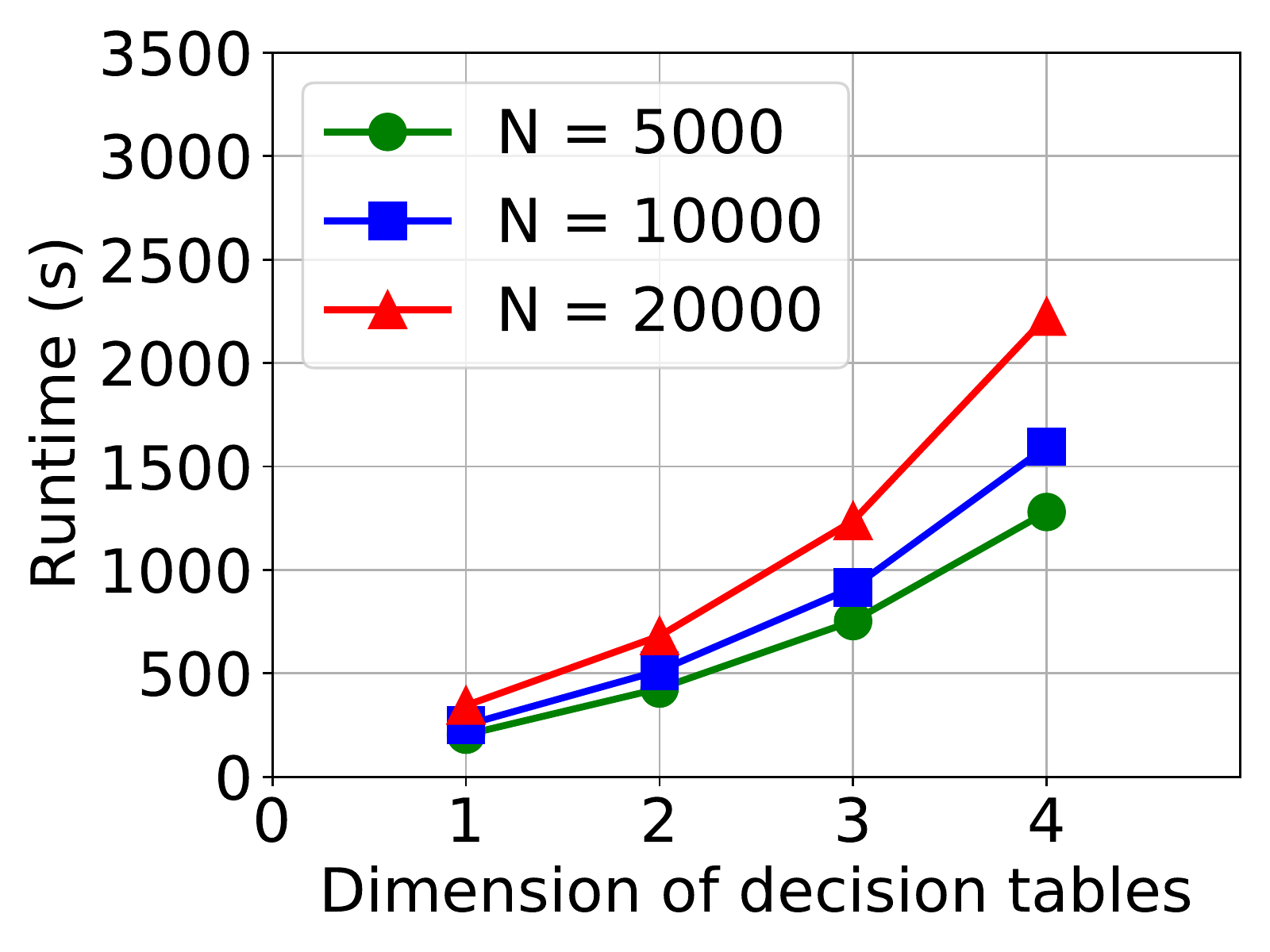}\\(a)
	\end{minipage}
	\begin{minipage}[t]{0.32\linewidth}
		\centering
		\includegraphics[width=\linewidth]{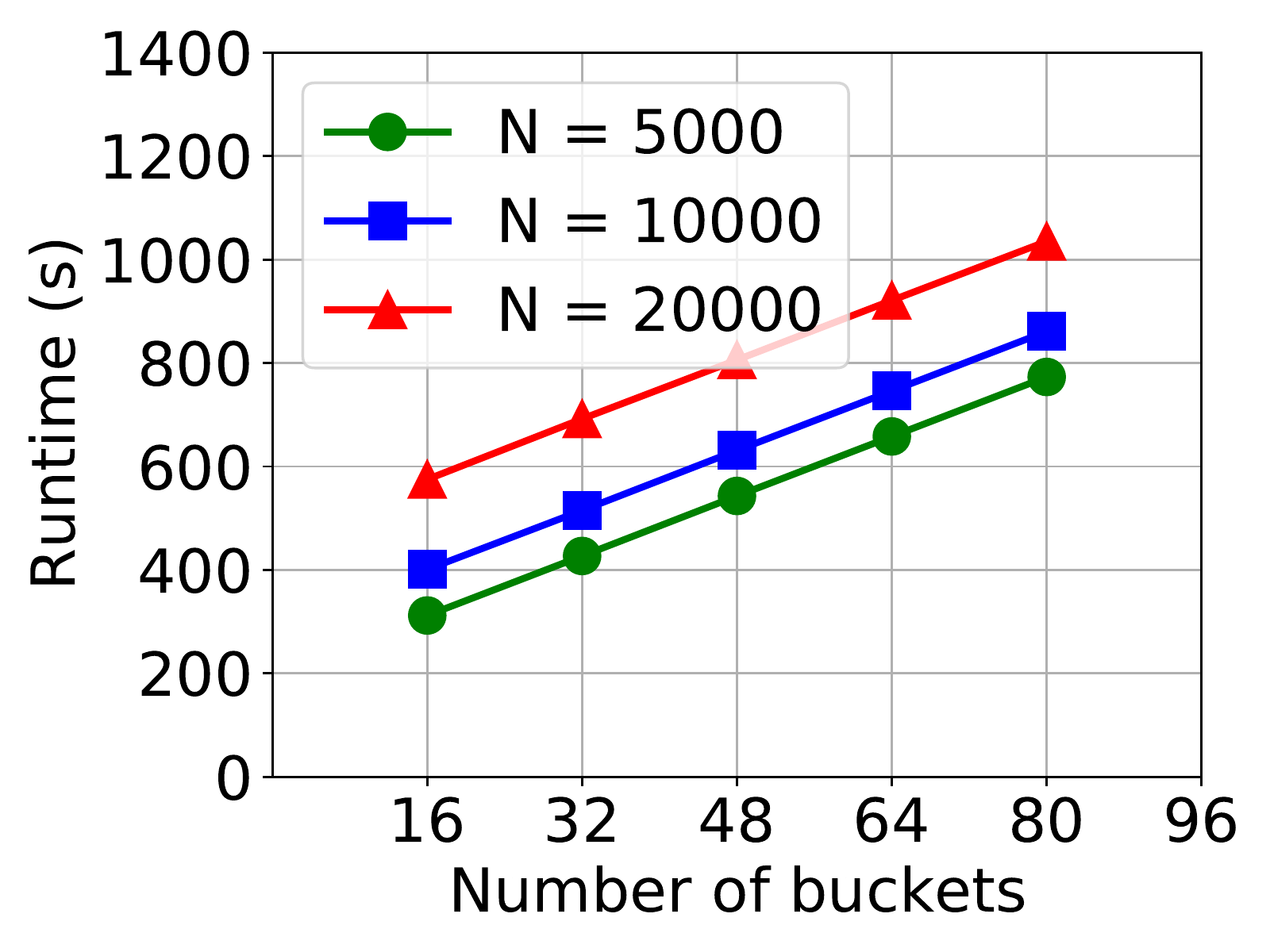}\\(b)
	\end{minipage}
	\begin{minipage}[t]{0.32\linewidth}
		\centering
		\includegraphics[width=\linewidth]{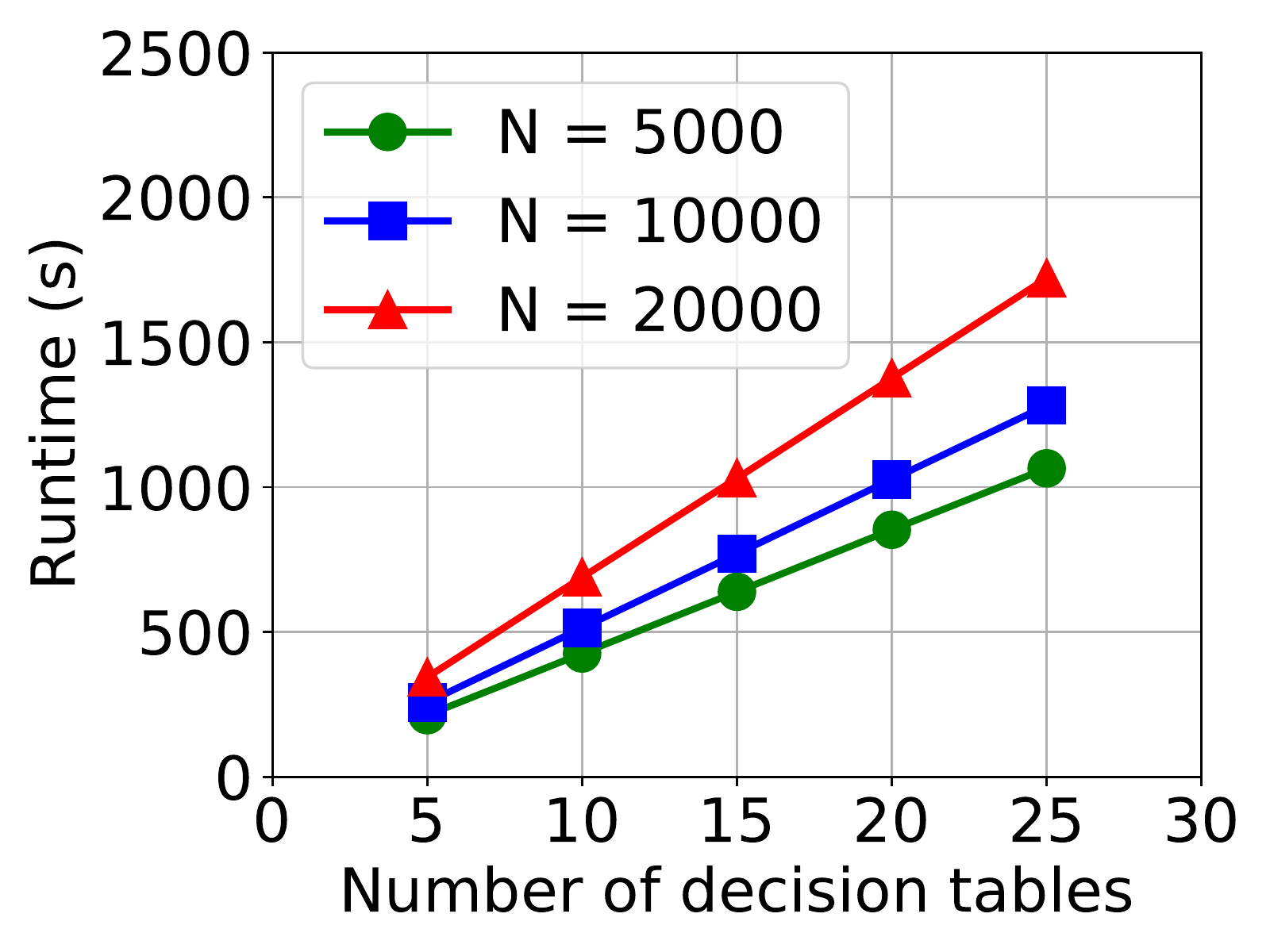}\\(c)
	\end{minipage}
	\caption{\revise{Runtime performance on Syn$^C$ with 40 features for different numbers of training samples $N$ and (a) varying dimension of decision tables $D$, (b) varying number of buckets $B$, and (c) varying number of decision tables $T$, respectively.}
	\label{fig:classi_scalability}}
	\vspace{-5pt}
\end{figure}

\subsection{Efficiency Evaluation}

We now report the computation and communication performance of {\main} in secure training over the three public datasets, and present the results in Table \ref{tab:efficiency}.
We also note that to evaluate the efficiency of {\main} over larger datasets, we use three synthetic datasets, of which the results are reported in Section \ref{sec:scala}.

From the results in the first two records of Table \ref{tab:efficiency}, we can observe that the training time and communication cost of {\main} on Cal Housing is significantly more than that on Credit, which is because we train 50 decision tables on Cal Housing but only 10 decision tables on Credit. 
However, {\main} consumes an average of 127 seconds to train a 5-dimensional decision table on Cal Housing and 172 seconds to train a 4-dimensional decision table on Credit. 
There are mainly two reasons for this observation: (1) the number of features and the number of samples of Cal Housing are less than that of Credit. (2) dealing with classification tasks additionally needs secure Sigmoid evaluation compared with regression tasks, which will result in more computation and communication overhead.

From the results in the last two records of Table \ref{tab:efficiency}, we can observe that although the scale of Credit is nearly 50 times larger than that of Breast Cancer, the training time on Credit under similar parameters is only roughly 4 times that on Breast Cancer. 
This is because {\main} securely discretizes training data into buckets, and all the time-consuming computations are conducted on the buckets. 
In this way, the training cost is highly correlated with the number of buckets, instead of the number of training samples. 
In fact, this also indicates the strong scalability of {\main}, which will be demonstrated in detail in the next section.

\subsection{Scalability Evaluation}
\label{sec:scala}
\revise{
We now evaluate the scalability of {\main}. 
To examine how the number of training samples and features affect the cost of secure online training, we conduct an experiment using the Syn$^A$ dataset with 20,000 samples and 80 features. 
We report the results in Fig. \ref{fig:runtime-commn-varying-sample} and Fig. \ref{fig:runtime-commn-varying-feature}, which show the runtime and communication cost for varying numbers of samples and features, respectively. 
For the experiment related to Fig. \ref{fig:runtime-commn-varying-sample}, we fix the number of features as 80, the dimension $D$ as 2, the number of buckets $B$ as 32, and the number of decision tables $T$ as 10, for varying number of training samples by randomly selecting samples from Syn$^A$. 
It is observed that the runtime is not much affected by the number of training samples from Fig. \ref{fig:runtime-commn-varying-sample} (a). At the same time, the communication cost grows linearly with the number of training samples as shown in Fig. \ref{fig:runtime-commn-varying-sample} (b). 
For the experiment related to Fig. \ref{fig:runtime-commn-varying-feature}, we randomly select 10,000 samples from Syn$^A$ and fix the dimension $D$ as 2, the number of buckets $B$ as 32, and the number of decision tables $T$ as 10, for varying number of features by randomly selecting features from Syn$^A$. 
From Fig. \ref{fig:runtime-commn-varying-feature} (a) and Fig. \ref{fig:runtime-commn-varying-feature} (b), we can observe that both the runtime and communication cost of {\main} increase proportionally with the number of features, in line with the complexity of Algorithm \ref{alg::sectable}, where the main loop in each level enumerates features.
}

\revise{
Next, we examine the impact of dimension $D$,  the number of buckets $B$, and the number of decision tables $T$ on the runtime of secure training. 
We use the synthetic regression dataset Syn$^B$ with 20,000 samples and 60 features for the regression task, and the synthetic classification dataset Syn$^C$ with 20,000 samples and 40 features for the classification task.
We first examine the relationship between dimension $D$ and runtime. 
For both regression and classification tasks, we set $B$ = 32, $T$ = 10, and vary the dimension $D$, over varying number of training samples. The results are shown in Fig. \ref{fig:reg_scalability} (a) and Fig. \ref{fig:classi_scalability} (a), from which we can observe that the runtime grows exponentially with the increase of dimension.
The results are consistent with the complexity of our secure training algorithm because the number of tree nodes is exponentially related to the dimension. The computation in each node accounts for the largest proportion of all calculations.
We then examine the relationship between the number of buckets $B$ and runtime.
We set $T$ = 10, $D$ = 2, and vary $B$, over varying number of training samples. 
The evaluation is also performed over both Syn$^B$ and Syn$^C$. 
Fig. \ref{fig:reg_scalability} (b) and Fig. \ref{fig:classi_scalability} (b) show the results, which indicate the linear association between $B$ and runtime. 
Finally, we evaluate the relationship between the number of decision tables $T$ and runtime, and summarize the results in Fig. \ref{fig:reg_scalability} (c) and Fig. \ref{fig:classi_scalability} (c).
We set $T$ from 5 to 25, while keeping $B$ = 32 and $D$ = 2, on both Syn$^B$ and Syn$^C$ with varying number of training samples.
Recall that in {\main} an ensemble of decision tables is securely built. The training time of each decision table is roughly the same when we set the same parameters for them.  So the runtime must grow linearly with the number of decision tables, which is consistent with the results in Fig. \ref{fig:reg_scalability} (c) and Fig. \ref{fig:classi_scalability} (c).
In summary, the above evaluation results demonstrate that Privet is scalable and capable of handling large-scale datasets with a large number of features and training samples.
}

\section{Conclusion}
\label{sec:conclusion}
In this paper, we design, implement, and evaluate {\main}, the first system framework enabling privacy-preserving VFL service for gradient boosted decision tables. 
Building on lightweight secret sharing techniques, {\main} supports an arbitrary number of distributed participants to collaboratively train gradient boosted decision tables over vertically partitioned distributed datasets, offering strong protection for individual data as well as for intermediate outputs.
Extensive experiments on several real-world datasets and synthetic datasets demonstrate that {\main} achieves promising performance, with model utility comparable to the case of plaintext centralized learning. 
For future work, it would be an interesting direction to explore the possibility of leveraging GPUs to achieve further performance boost.

\section*{Acknowledgement}

This paper was supported in part by the Guangdong Basic and Applied Basic Research Foundation under Grants No. 2021A1515110027, No. 2023A1515010714, and No.  2021A1515011406, by the Shenzhen Science and Technology Program under Grants No. RCBS20210609103056041 and No. JCYJ20220531095416037, by the National Natural Science Foundation of China under Grant No. 62002167, and by the Natural Science Foundation of JiangSu Province under Grant No. BK20200461.

\bibliographystyle{IEEEtran}
\bibliography{ref}

\end{document}